\documentclass[11pt,letterpaper]{article}
\usepackage[top=80pt,bottom=80pt,left=88pt,right=86pt]{geometry}
\usepackage{amssymb}
\usepackage{graphicx}
\usepackage[ruled,vlined,linesnumbered]{algorithm2e}
\usepackage{microtype}
\usepackage{subfig}
\usepackage{amsmath}
\usepackage{amsthm}
\usepackage{enumerate}
\usepackage{wrapfig}
\usepackage{multirow}
\usepackage{nicefrac}
\usepackage[colorlinks,linkcolor=blue,filecolor=blue,citecolor=blue,urlcolor=blue,pdfstartview=FitH,pagebackref]{hyperref}
\usepackage[nameinlink]{cleveref}
\usepackage{authblk}
\usepackage[dvipsnames]{xcolor}

\newcommand{\Nats}{{\mathbb{N}}}             
\newcommand{\Reals}{{\mathbb{R}}}            
\newcommand{\eps}{\varepsilon}               

\newcommand{\FSD}{\mathsf{FSD}_{\delta}}

\newcommand{\spine}{\mathsf{SP}_{\delta}}
\newcommand{\Spine}{\mathsf{SP}}
\newcommand{\slice}{\mathsf{SL}_{\delta}}
\newcommand{\Slice}{\mathsf{SL}}

\newcommand{\distance}{\mathsf{D}}
\newcommand{\Gdistance}{\overrightarrow{\delta_G}}

\newcommand{\WGdistance}{\overrightarrow{\delta_{wG}}}
\newcommand{\Frechet}{\mathsf{F}}

\newcommand{\wFrechet}{\mathsf{wF}}

\newcommand{\Traversal}{\overrightarrow{\mathsf{\delta_T}}}

\usepackage{mathtools}
\DeclarePairedDelimiter\segment {\langle}{\rangle}

\newcommand{\spinal}{spinal~}

\newcommand{\Prnt}{\mathsf{Parent}}
\newcommand{\true}{\mathsf{T}}
\newcommand{\false}{\mathsf{F}}
\newcommand{\Cld}{\mathsf{Children}}
\newcommand{\enq}{\mathsf{Enqueue}}
\newcommand{\deq}{\mathsf{Dequeue}}
\newcommand{\pop}{\mathsf{Pop}}
\newcommand{\push}{\mathsf{Push}}
\newcommand{\weight}{\mathsf{W}}

\def\D{{\cal D}}
\def\E{{\cal E}}
\def\F{{\cal F}}

\def\I{{\cal I}}

\def\P{{\cal P}}

\def\R{{\cal R}}

\def\T{{\cal T}}

\def\V{{\cal V}}

\newcommand{\figref}[1]{Figure~\ref{#1}}

\usepackage{url}

\newcommand{\myemph}[1]{\emph{#1}}
\newcommand{\mathemph}[1]{#1}

\newtheorem{theorem}{Theorem}
\newtheorem{lemma}[theorem]{Lemma}

\newtheorem{observation}[theorem]{Observation}

\newtheorem{definition}[theorem]{Definition}

\newtheorem{remark}{Remark}

\def\eps{{\varepsilon}}

\newcommand{\frechet}{Fr\'echet}

\newcommand{\payAttention}[3]{{[[{\color{#1}{\textbf{#2: }}}{\color{blue}{#3}}]]}}
\newcommand{\xSays}[3]{\payAttention{#1}{#2 Says}{#3}}
\newcommand{\carola}[1]{\xSays{Purple}{Carola}{#1}}
\newcommand{\majid}[1]{\xSays{NavyBlue}{Majid}{#1}}

\usepackage[colorinlistoftodos,prependcaption,textsize=tiny]{todonotes}

\usepackage[numbers,sort]{natbib}

\definecolor{lgray}{gray}{0.5}
\usepackage{diagbox}
\usepackage{makecell}
\usepackage{comment}
\title{Minimum-Complexity Graph Simplification under Fr\'echet-Like Distances\footnote{This project has been supported by NSF grant (AitF: NSF-CCF 1637576). 
	}
}
\author{Omrit Filtser\thanks{Department of Applied Mathematics and Statistics, State University of New York at Stony Brook, NY. U.S., omrit.filtser@gmail.com}~~~
	Majid Mirzanezhad\thanks{Transportation Research Institute, College of Engineering, University of Michigan, Ann Arbor, MI. U.S., miirza@umich.edu}~~~
	 Carola Wenk \thanks{Department of Computer Science, Tulane University, New Orleans, LA. U.S., cwenk@tulane.edu}}
\date{}
\begin{document}

\maketitle
\begin{abstract}
	
Simplifying graphs is a very applicable problem in numerous domains especially in computational geometry. 
Given a geometric graph and a threshold, the minimum-complexity graph simplification asks for
computing an alternative graph of minimum complexity
so that the distance between the two graphs remains at most the
threshold. In this paper we propose several NP-hardness and
algorithmic results depending on the type of input and simplified
graphs, the vertex placement of the simplified graph, and the distance measures
between them (graph and traversal distances~\cite{abksw-19, aerw-mpm-03}). 
In general, we show that for arbitrary input and output graphs, the
problem is NP-hard under some specific vertex-placement of the simplified graph. 
When
the input and output are trees, and the graph distance is applied from the
simplified tree to the input tree, we give an $O(kn^5)$ time algorithm,
where $k$ is the number of the leaves of the two trees that are identical and $n$ is the number of vertices of the input.
\end{abstract}
\section{Introduction}\label{sec:intros}

Unlike curve simplification problem, simplifying structurally more complicated input objects such as trees and graphs has not been extensively studied in the computational geometry community. This problem may have applications in GIS, image processing, shape analysis, mesh simplification, molecular biology, etc.~\cite{glw-mpstd-07, l-hdrknn-91, cms-cmsa-98, bs-rnaspsaknn-06}. 
In  a generic application, a user may wish to obtain a coarse and simpler representation of a map preserving the geometry of the underlying structure. This can bring the idea of computing an alternative graph with minimum-complexity to the scene of simplification. 
There are a few works that study approximating a planar subdivision of a map (plane graph) with the minimum number of links under some topological constraints \cite{estkowski01simple, ghms-apswmlp-93}. Most of the algorithms have applied to GIS data and are based on map schematization \cite{m-smcs-14} in which, roughly speaking, the main topological structure of the map remains the same and paths with vertices of degree two become simplified. A generalization of map schematization under some topological constraints, e.g., facet preserving, no self-intersecting boundary simplification can be found in \cite{funke2017map,m-apsstcep-18}.

There are a few works considering the simplification of a given planar subdivision inside a polygonal region with another minimum-link planar subdivision homeomorphic to the original one. In this setting the input is a plane graph and the problem is more concerned with topological constraints of the subdivision inside the polygon \cite{estkowski01simple,ghms-apswmlp-93}. Most of the optimization problems on simplifying graphs fall into NP-hard or APX-hard classes of problems. In \cite{estkowski01simple} a heuristic algorithm for planar maps is proposed that keeps the boundary of  polygonal regions simple after simplification while it is impossible to give a polynomial-time algorithm within $n^{0.2-\epsilon}$ approximation factor for the problem, for any $\epsilon>0$, assuming $\textsf{P} \neq \textsf{NP}$.  

In the context of simplification, one can consider different variants of the (graph) simplification problem induced by the input parameters/constraints such as the type of the distance measure, e.g., Fr\'echet, Hausdorff distances, the direction (if the distance is  asymmetric), topological constraints (facet preservation), and others. Note that some distances are not symmetric, therefore it is crucial in which direction the distance is applied. 
We start with defining the problem setting, and we propose algorithmic as well as NP-hardness results for the problem. Our objective is to study the minimum complexity simplification problem for graphs and examine the difficulty of the problem for different types of input and simplified output, i.e., trees, and graphs. Suppose we are given a positive real number and an input graph that can be either a tree or graph in general. We are interested in approximating the input graph using another graph with the minimum complexity, where the distance from the input to the output (or the opposite) is at most the given threshold. 
We call this generic problem the \textsc{Minimum-Complexity Graph Simplification} (\textsc{MCGS}). 
Here, the input is meant to be more complex than the output object in terms of the structure. For example, the output of the \textsc{MCGS} with an input tree may not admit a graph (with cycles) but either a tree or a path. 
We define the problem for graphs but it can apply to trees as well:  

\paragraph{The \textsc{MCGS} Problem:}
Let $\delta>0$ be a real value, $\distance(\cdot,\cdot)$ be a distance
measure between graphs, and $G=(V,E)$ be a (connected) graph in $\Reals^d$, whose edges
in $E$ are straight-line segments between the vertices in $V$. We aim
to compute an alternative (connected) graph $G'$, with the minimum-complexity satisfying $\distance(G, G')\leq \delta$.

\section{Classification of the Problem}

Our objective is systematically go over different combinations of the input/output graphs, vertex restrictions and distance measures. By \emph{`minimum-complexity'}, we consider minimizing the total number of edges and vertices of a connected graph. 
We use two Fr\'echet-like (directed) distances between graphs; traversal and graph distances, in particular. Under a \emph{graph mapping}, a graph is mapped continuously to a portion of  the other, in such a way that edges are mapped to paths in the other graph. The \emph{graph distance} is then defined as the maximum of the Fr\'echet distances between the edges and the paths they are mapped to. The \emph{traversal distance}  converts graphs into curves by traversing the graphs continuously and comparing the resulting curves
using the Fr\'echet distance. In other words, it compares the traversal of a man on a graph with the traversal of his dog on part of the other graph while staying close to each other~\cite{abksw-19}.

For the vertex placement of  the simplified graph, we differentiate between \emph{vertex-restricted}, \emph{edge-restricted}, and \emph{non-restricted} variants. In the vertex-restricted case, a simplified graph selects its vertices from a subset of the input vertices, while in the edge-restricted case it selects a subset of points on any edge of the input graph. If a simplified graph selects its vertices from anywhere in the ambient space then it is a non-restricted simplification. A \emph{subgraph-restricted} simplification is a special case of the vertex-restricted setting in which the simplified graph is a subgraph of the input graph. Another special case of the vertex-restricted setting is the \emph{leaf-restricted} simplification which requires the degree-one vertices of the output to be identical to a subset of degree-one vertices of the input graph. This might be helpful in capturing the structure of the graph (see e.g.~\cite{abksw-19}). 
See Figure~\ref{fig:variants} to understand the relationships between the restrictions. We will explain the formal definition on this in the dedicated sections.
\begin{figure}[htbp]
	\centering
	\includegraphics[width=.6\textwidth]{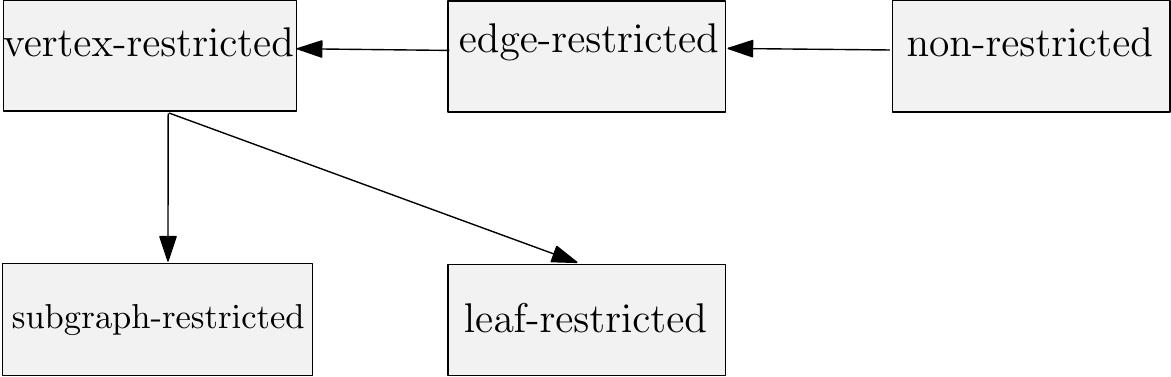}
	\caption[Different variants of the MCGS and their relationships]{Three primary simplification variants above contain the two other secondary variants as depicted. An arrow from a variant $A$ to a variant $B$ of the problem implies that any solution to $B$ can be a solution to $A$ as well.}
	\label{fig:variants}
\end{figure}

In principle, we call a variant of \textsc{MCGS} with restriction $\R$ on the placement of vertices of the simplified graph that simplifies $G_1$ to a minimum-complexity simplified graph $G_2$, applying the distance $\overrightarrow{\distance}(G_2,G_1)$ (from output $G_2$ to input $G_1$), an \emph{$\R$-restricted min-complexity $G_1$-$G_2$ simplification under $\overrightarrow{\distance}(G_2,G_1)\leq \delta$}. 
See Figure~\ref{fig:differentsimp} for better understanding of different simplifications from input to output. 
\begin{figure}[htbp]
	\centering
	\includegraphics[width=0.65\textwidth]{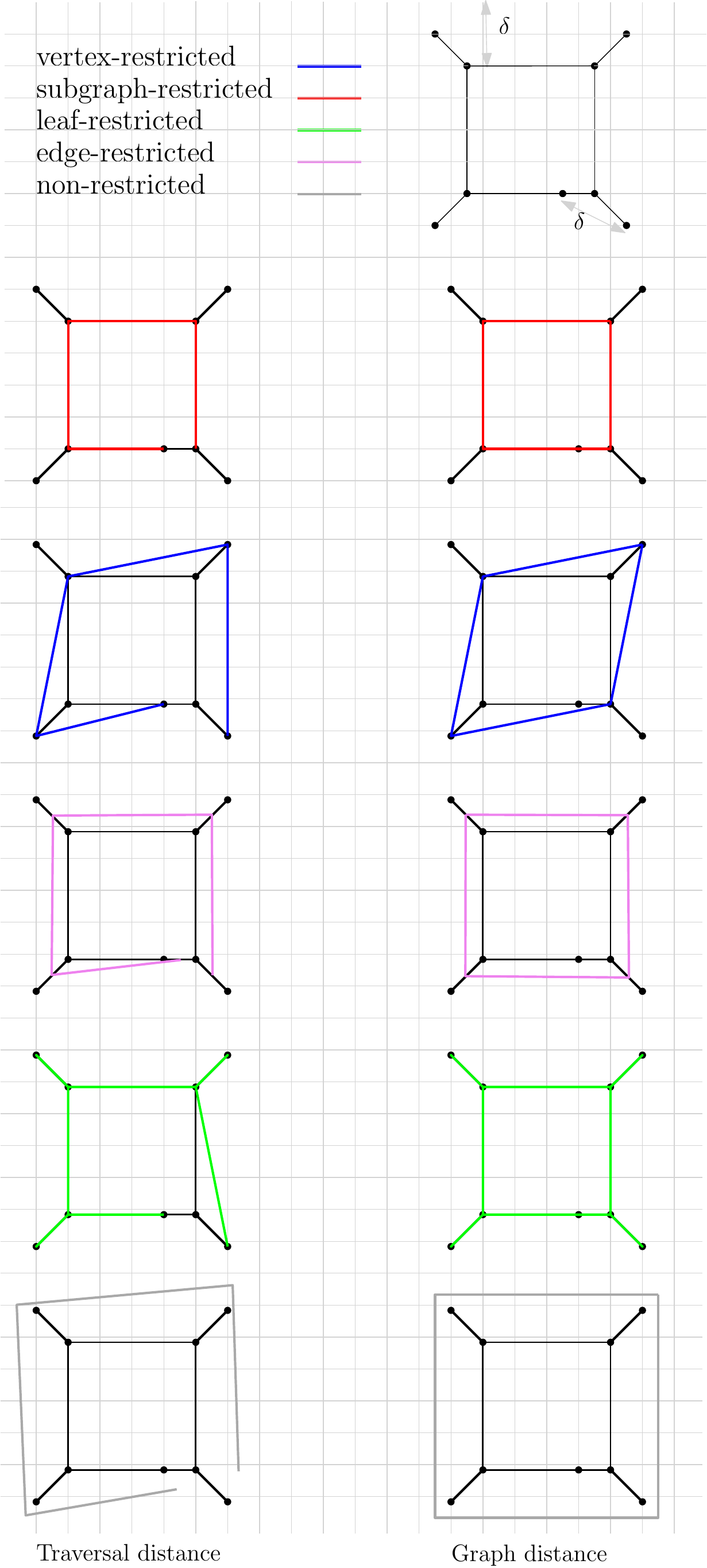}
	\caption[Minimum-Complexity simplification of a graph with different restrictions]{Minimum-Complexity simplification of a graph with different restrictions. Under the traversal distance the simplified graphs are curves in this example.}
	\label{fig:differentsimp}
\end{figure}

As mentioned earlier, the \textsc{MCGS} problem is clearly a very generic problem. The hardness of a problem variant or the efficiency of the algorithm for that variant depends critically on the choice of the distance $\distance$ between the input and the output as well as the vertex-placement restriction. 
We wish to consider a certain type of distance measures between graphs that extends the minimum-link simplification problem under the Fr\'echet distance for curves in~\cite{kklmw-gcs-19} to the one for graphs. 
While our main concern in this work is to take the geometry between the input and output graphs into account, by changing the input and output over graphs and trees we somehow give the user the choice of retaining the topology. This way we may control the topology between the input and output graphs unlike the other existing works that propose algorithmic treatment for maintaining the topology~\cite{m-apsstcep-18, estkowski01simple}.   


There are various distances considered between graphs and trees in the literature such as ``Graph edit distance''~\cite{jh-ablpfged-06, cgkss-msgg-08}, and ``Contour tree distance''~\cite{bos-ctd-17}. The former does not respect the continuity of  the curves and the latter is also a generalization of the Fr\'echet distance to graphs, but NP-hard to compute between them. In this paper, we focus on two Fr\'echet-like distance measures between graphs and/or tree; the \emph{graph distance} proposed in~\cite{abksw-19} and \emph{traversal distance} in~\cite{aerw-mpm-03}. As shown in~\cite{abksw-19}, the traversal distance is not greater than the graph distance. Throughout the paper,  we use the term ``graph distance'' to refer to both weak and strong types, unless we specifically mention the type of the distance. A comparison of the traversal distance and graph distance can be found in Figure~\ref{fig:dist-comparison}. 
\begin{figure}[htbp]
	\begin{center}
		\includegraphics[width=.30\textwidth]{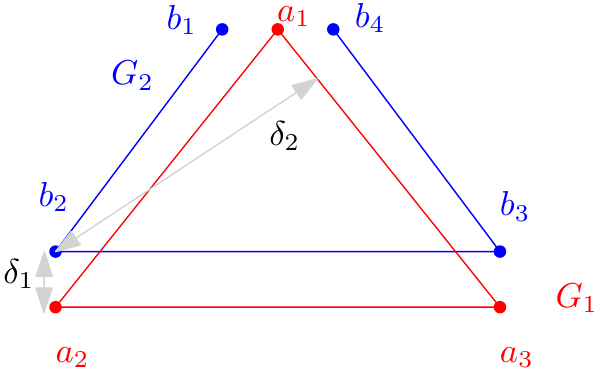}
		\caption[Comparison between the traversal and graph distances]{Under the optimal mapping $\mu:G_1 \rightarrow G_2$, where $\mu(a_1) = b_1,~\mu(a_2)=b_2$, and $\mu(a_3)= b_3$, we have the graph distance equal $\delta_2$ while the traversal distance equal $\delta_1$ with $\delta_2 \gg \delta_1$.}
		\label{fig:dist-comparison}
	\end{center}
\end{figure}

\subsection{Our results}
Inspired by globally simplifying a curve in~\cite{kklmw-gcs-19} we have restrictions on the placement of the vertices of the simplified graph.  In this paper, we primarily study two restrictions; vertex- and edge-restricted. We first show that the vertex-restricted min-complexity tree-tree simplification under the traversal distance from input to output is NP-hard (Theorem~\ref{thm:tree-tree-traversal}). Although the NP-hardness of the same variant under the graph distance remains elusive, we give a fixed-parameter tractable algorithm that runs in $O\big(n^3\alpha^2 2^\alpha\big)$ time and $O\big(n^5+n^3\alpha^2 2^\alpha\big)$ time under the weak and strong graph distances, respectively, from  input to the output tree. Here, $\alpha $ is an implicit parameter that, roughly speaking, is the number of  intersections between the simplified tree and the ball of radius $\delta$ around each vertex of the input tree (Theorem~\ref{thm:fixedTTGD}). 
\begin{table}[htbp]
	\small
	\centering
	\begin{tabular}{|l||*{3}{c|}}\hline
		\backslashbox{Restriction}{Distance}
		&\makebox{Graph distance}&\makebox{Traversal distance}&\makebox{Assumptions} \\\hline\hline
		
		Vertex-restricted & \makecell{$O\big(n^5+n^3\alpha^2 2^\alpha\big)$ \\ (Thm.~\ref{thm:fixedTTGD})} &\makecell{NP-hard \\ (Thm.~\ref{thm:tree-tree-traversal})}&\makecell[l]{ Tree-to-Tree  \\ Input $\rightarrow$ Output} \\\hline
		
		Edge-restricted & \makecell{(Weakly) NP-hard \\ (Thm.~\ref{thm:edge-restricted-gg}) }& \makecell{(Weakly) NP-hard \\ (Thm.~\ref{thm:edge-restricted-gg}) } &\makecell[l]{ Graph-to-Graph \\ Tree-to-Tree \\ Input $\rightarrow$ Output }\\\hline 
		
		Subgraph-restricted  & \makecell{NP-hard \\ (Thm.~\ref{thm:subgraph-graph-graph})} & ?
		&\makecell[l]{ Graph-to-Graph \\ Input $\rightarrow$ Output}\\\hline
		\multirow{3}{*}{Leaf-restricted} & \makecell{NP-hard \\ (Thm.~\ref{thm:leaf-restricted-tree})} &\makecell{NP-hard \\ (Thm.~\ref{thm:leaf-restricted-tree})}&\makecell[l]{ Tree-to-Tree \\ Input $\rightarrow$ Output\\~}\\
		& \makecell{$O(kn^5)$ \\ (Thm.~\ref{thm:OTTLGD})} & ? &\makecell[l]{ Tree-to-Tree \\ Output $\rightarrow$ Input  } \\\hline
		
	\end{tabular}
	\caption[Our results on  the \textsc{MCGS} problem]{\label{tbl:results} Our results on  the \textsc{MCGS} problem. 
	}
\end{table}
As the vertex-restricted variant appears to be hard to admit a  fully polynomial time algorithm we take our investigation further to two related variants; subgraph- and leaf-restricted ones. 
We show that the subgraph-restricted minimum-complexity graph-to-graph simplification under the graph distance
from input to output is also NP-hard (Theorem~\ref{thm:subgraph-graph-graph}).  
We show that the leaf-restricted tree-tree simplification under both graph and traversal distances from input to output is NP-hard (Theorem~\ref{thm:leaf-restricted-tree}).
However, when the direction of the distance changes from output to input, while the leaves of the output must be identical  to $k$ leaves of the input tree, we propose an $O(kn^5)$ time algorithm that uses $O(kn^2)$ space (Theorem~\ref{thm:OTTLGD}). 
Although the two latter variants might be considered special cases, still the difficulty of the  problem does not significantly change regardless of the distance measure we use. We take this investigation as part of our systematic study in this work and we believe that one can define other variants for which the problem admits efficient polynomial time algorithms.
In the end, by a modification of the construction in~\cite{kklmw-gcs-19} we show that the edge-restricted min-complexity graph-graph and tree-tree simplification under the graph and traversal distances becomes (weakly) NP-hard. This weakly NP-hardness result holds for all types of input and output graphs, i.e., graphs, trees, and curves (Theorem~\ref{thm:edge-restricted-gg}). See a summary of our results in Table~\ref{tbl:results}.

\begin{remark} Although tree is a special case of  graph, the NP-hardness for $\R$-restricted min-complexity graph-graph simplification does not immediately follow from the one for $\R$-restricted min-complexity tree-tree simpiifcation since these two problems do not have the same type of output. The former outputs a graph but the latter necessarily outputs a tree.  
\end{remark}

\section{Preliminaries}\label{sec:prelim}
We first begin with introducing curves and the Fr\'echet distance between them. Let $P=\langle p_1,p_2,\cdots, p_n \rangle$ be a polygonal curve.
We treat $P$ as a continuous
map $P:[1,n] \rightarrow \mathbb{R}^d$, where $P(i)=p_i$ for an integer
$i$, and the $i$-th edge is linearly parameterized as $P(i + \lambda) =
(1-\lambda) p_i + \lambda p_{i +1}$.
We write $P[s,t]$ for the subcurve between
$P(s)$ and $P(t)$ and $\segment{P(s)P(t)}$ for the line segment connecting the two points. 
Given two curves $P:[1,n]\rightarrow \mathbb{R}^d$ and $Q:[1,m]\rightarrow \mathbb{R}^d$, the \frechet\ distance between $P$ and $Q$ is defined as:

$$\Frechet(P,Q)=\inf_{f,g}\max_{t\in[0,1]}\Vert P(f(t))-Q(g(t))\Vert ,$$
where $f:[0,1]\rightarrow [1,n]$ and $g:[0,1]\rightarrow [1,m]$ are continuous non-decreasing functions. If $f$ and $g$ are not  non-decreasing functions, then the obtained distance is called \emph{the weak Fr\'echet distance}  denoted by $\wFrechet(P,Q)$. The \emph{free space diagram} of the two curves $P$ and $Q$ of complexities of $n$ and $m$, respectively, is denoted by $\FSD(P,Q)$. This diagram  consists of
$(n-1)\times (m-1)$ cells in the domain $[1,n]\times[1,m]$.   

For any $\delta>0$, the free space diagram $\FSD(P,Q) $ consists of \emph{cells} and the boundary of each cell consists of four sides which each contains at most one \emph{free space interval}. An interval is part of an edge of $P$ that is within $\delta$ to a vertex of $Q$ and vice versa. A \emph{monotone path} from $(1,1)$ to $(n,m)$ that lies entirely within the free space
corresponds to a pair of monotone re-parameterizations
$(f,g)$ that witness $\Frechet(P, Q)\leq\delta$. Alt and
Godau showed that such a \emph{reachable path} can be computed in
$O(mn)$ time by propagating reachable points across free space cell
boundaries in a dynamic programming
manner \cite{ag-cfdb-95}.
%
See \figref{fig:gcfs} for an example of free
space diagram for two curves.
\begin{figure}[!t]
	\begin{center}
		\includegraphics[width=.75\textwidth]{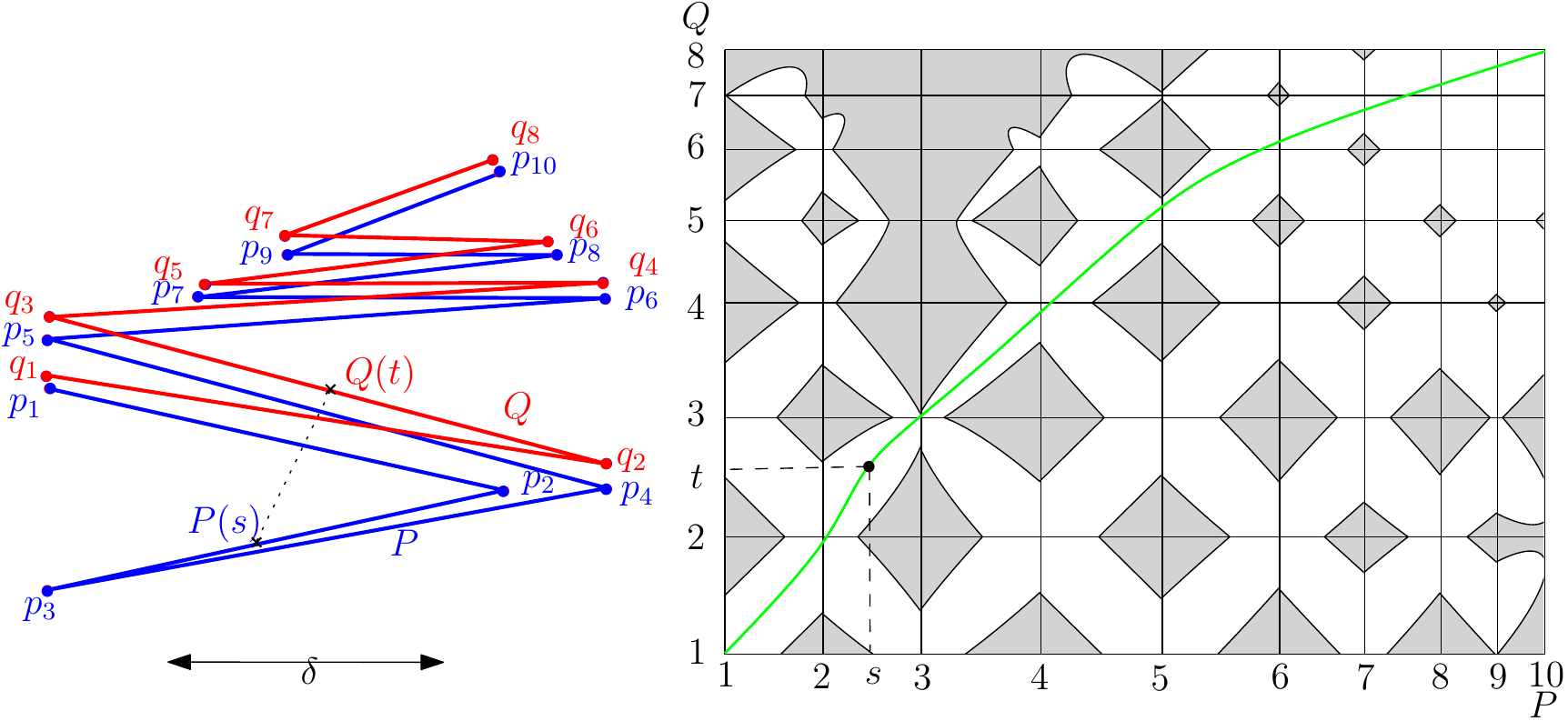}
		\caption{$(s,t)$ is a free point on a reachable path in $\FSD(P,Q)$  The free space is shown in white.}
		\label{fig:gcfs}
	\end{center}
\end{figure}

Now let $G = (V,E)$ be a graph \emph{immersed} in $\Reals^d$, where $V(G) = \{v_1,v_2,\cdots,v_n\}$, $v_i$ is embedded at a point $p \in \Reals^d$ for all $i\in \{1,\cdots,n\}$, and each edge $e=\segment{uv}\in E(G)$ is also linearly parameterized.
A continuous mapping $f:[0,1]\rightarrow G$ is called a \emph{traversal} of graph $G$ if it is surjective, and a \emph{partial traversal} of $G$ if it is not necessarily surjective. Given two graphs $G_1,G_2$ immersed in $\Reals^d$, their \emph{traversal distance} (introduced in \cite{aerw-mpm-03}) is:

$$
\Traversal(G_1,G_2)=\inf_{f,g}\max_{t\in[0,1]}\Vert f(t)-g(t)\Vert ,
$$
where $f:[0,1]\rightarrow G_1$ is a traversal of $G_1$, and $g:[0,1]\rightarrow G_2$ is a partial traversal of $G_2$. 

Suppose $\mu$ is an arbitrary mapping from graph $G_1$ to graph $G_2$ that maps every vertex of $G_1$ to some point on $G_2$. In other words $\mu(v) = p$ for all $v \in V(G_1)$ and $p \in G_2$. Given an edge $\segment{uv} \in E(G_1)$, with a slight abuse of the notation, we have $\mu(e) = \P$, where $e = \segment{uv}$, and $\P$ is a path starting at $\mu(u)$ and ending at $mu(v)$ on $G_2$. Now a \emph{graph mapping} (introduced in \cite{abksw-19}) is
 a function $\mu: G_1 \rightarrow G_2$ that
(1) maps each vertex $v\in V(G_1)$ to a point $\mu(v)$ on an edge of $G_2$, and (2) maps each edge $e  = \segment{uv}\in E(G_1)$ to a simple path $\mu(e)$, from $\mu(u)$ to $\mu(v)$, in $G_2$. 
The \emph{directed (strong) graph distance} between $G_1$ and $G_2$ is:
$$\Gdistance(G_1,G_2)=\inf_{\mu: G_1 \rightarrow G_2}\max_{e\in E(G_1)} \Frechet(e,\mu(e)).$$ Note that the path between $\mu(u)$ and $\mu(v)$ in $G_2$ may not exist which results in $\Frechet(e,\mu(e)) = \infty$. This case can occur if $G_2$ is disconnected.
The \emph{directed (weak) graph distance}  denoted by $\WGdistance$ is  obtained by replacing the $\Frechet$ with $\wFrechet$ in the definition.

\section{Vertex-Restricted  Tree-Tree Simplification under the Traversal Distance}\label{sec:vertexrestricted}
In this section we show that  the  vertex-restricted min-complexity tree-tree
simplification under the traversal distance from input to output is NP-hard. Our reduction is from the \emph{minimum dominating set of unit disk graph} (\textsc{MDSUDG}) problem: given a unit disk graph $G=(V,E)$ in the plane  with $P = V(G)$ and $n = |P|$,
the \textsc{MDSUDG} problem asks for a set $S \subseteq P$ of minimum size such that every vertex in $P-S$ is adjacent to at least one vertex in $S$. The 
\textsc{MDSUDG} problem is known to be NP-hard~\cite{ccj-udg-90}.  Our reduction
takes an instance of \textsc{MDSUDG} and converts it to $T$ in $\Reals^2$ that is somewhat a star graph:
Let $\delta = 1$ and let $B$ be the smallest axis-parallel box that contains all unit disks around vertices in $P=\{p_1,\cdots, p_n\}$. Let $P'=\{p'_1,\cdots,p'_n\}$ be the set of vertices that are obtained by vertically translating the vertices in $P$ upward. The obtained translated vertex set is denoted by $P'$. Note that the length of the translation dominates $1$ significantly. We draw a straight-line segment between every vertex in $P$ and its corresponding translated vertex in $P'$. Aside from the vertices in $P\cup P'$ we have some other vertex $w$ where every vertex in $P'$ is connected through a path to it. Such a path is called \emph{bottleneck path}. We similarly define box $B'$ with respect to the point set $P'\cup \{w\}$.  For the rest of the reduction we have the following construction in box $B'$:
\begin{figure} [!h]
	\begin{center}
		\includegraphics[width=0.7\textwidth]{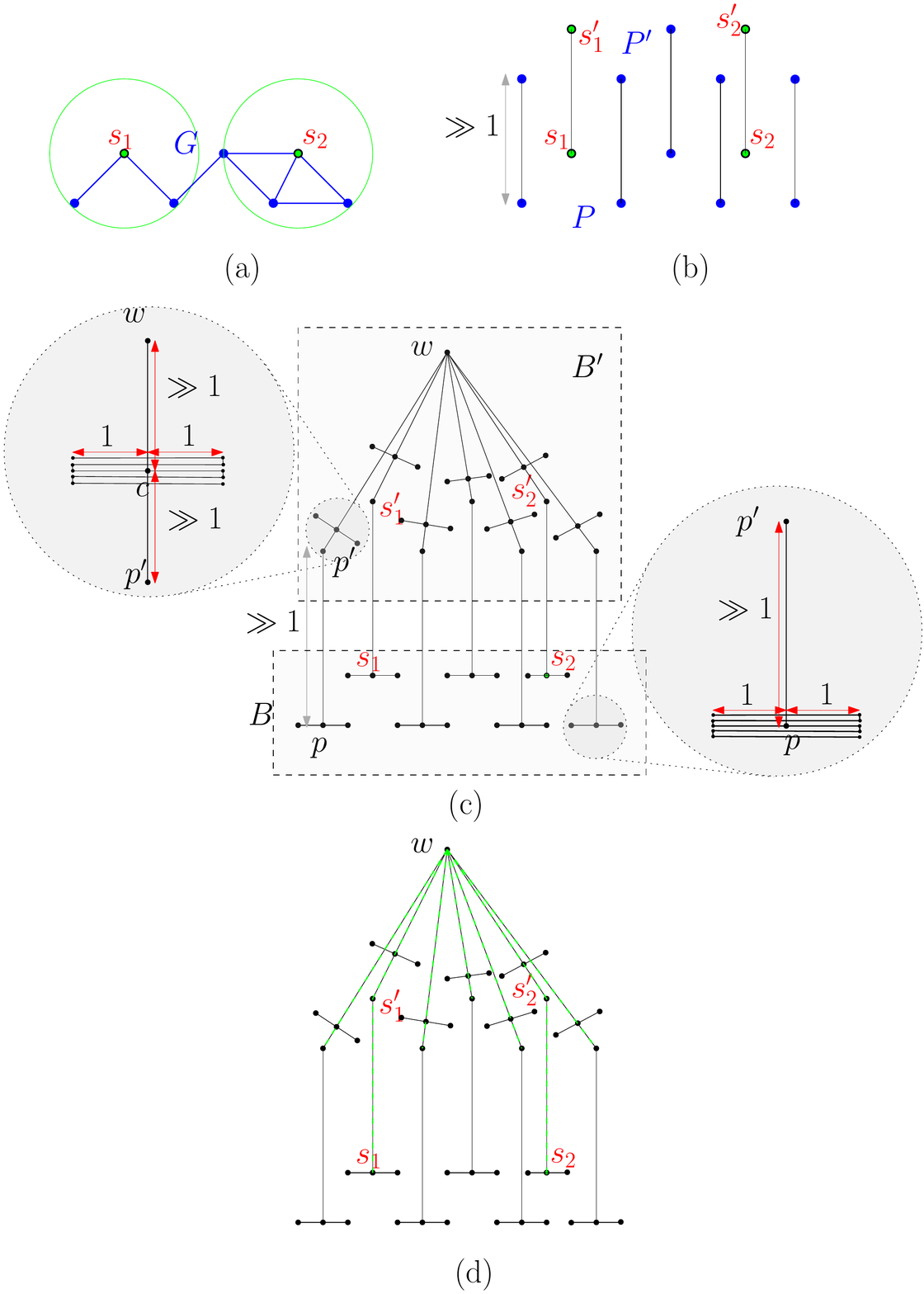}
	\end{center}
	\caption[The reduction tree for the vertex-restricted variant]{(a) The unit disk graph $G$. Here $\{s_1, s_2\}$ is a minimum dominating set for $G$. (b) The original point set $P$ at the bottom, which is identical to the vertex set of $G$, and the translated point sets $P'$ at the top. (c) The reduction tree, and the bottleneck paths. Any simplified link should pass through the center of the twist (the edges of the zigzag edges are distorted vertically for a better presentation). (d) $T$ is presented in black solid line, and $T'$ shown in green dashed line.}
	\label{fig:nphardness2d}
\end{figure}
\paragraph{Bottleneck path:}
A path consisting of a \emph{straight} edge between $p' \in P'$ and $w$, and a set of edges {prependicular} onto the straight-line edge that are called \emph{zigzag} edges. The zigzag edges and straight edge intersect each other at a vertex called the \emph{center} $c$ of the bottleneck path, thus the straight edge is broken down into two  edges, i.e., $\segment{p'c}$ and $\segment{cw}$. The length of the zigzag  edges is 2, where the left endpoint of the edge is at distance 1 to the center and the right endpoint of it is at distance 1 to the center as well. The length of each straight edge $\segment{p'c}$ and $\segment{cw}$ is significantly larger than 1. 

The center $c$ is clearly a vertex of degree 4 (see Figure~\ref{fig:nphardness2d}). 
This way any link simplifying the bottleneck path between $p'_i$ and $w$, has to pass through the center of the path and does not simplify another bottleneck path between $p'_j$ and $w$ at the same time with $j\neq i$. In other words, a simplified tree has exactly $n$ links connecting $w$ to all $p'_i$. The only remaining part is to place the centers and $w$ in a way that none of the pairs of paths ending at $w$ overlay onto each other and therefore cause to simplify multiple bottleneck path using one link. This can be done in polynomial time by not locating $w$ on the line supporting $p$ and $p'$ for all $p, p'\in P\cup P'$. Finally we place the same zigzag edges on each $p \in P$. This restricts our simplified tree's leaves  not to end at the leaves of $T$ but at the center of the zigzag edges that are originally the vertices in $P$ . 
The following theorem proves that there exists a dominating set of size $k$ if and only if there exists a simplified tree with $n+k$ edges.

\begin{theorem} \label{thm:tree-tree-traversal}
	Let  $T$  be a tree in $\Reals^2$ with $n$ vertices, and $\delta>0$. Computing a vertex-restricted min-complexity tree-tree simplification $T'$ under $\Traversal(T,T')\leq \delta$ is NP-hard.
\end{theorem}
\begin{proof}
	By construction $\delta=1$. Let the man walk on $T$ and his dog walk on the prospective  $T'$ such that they stay within distance 1 from each other which implies that $\Traversal(T,T')\leq 1$. Recall that  the man should traverse the entire $T$. Let $k>0$ be an integer as the decision parameter to the decision version of \textsc{MDSUDG}, and suppose there exists a dominating set $S=\{s_1,\cdots,s_m\}$ of the vertices in a unit disk graph $G$, i.e., $P=V(G)$. Analogously, let $S'=\{s'_1,\cdots,s'_m\}$ be the dominating set of the unit disk graph induced by vertices in $P'$. In other words, there is a set of unit disks $\D=\{D_1,\cdots, D_m\}$  covering the entire points in $S$ where the centers of these disks are the vertices in $S$. 
	
 Suppose there is a dominating set of size at most k, i.e., $m \leq k$. We show there exists a min-complexity simplified tree $T'$ whose number of edges is at most $n+k$.  We break down the man and dog's walks into the two following stages:
	
	(i) Box $B'$: By construction $T'$ has to select its edges as $\segment{wp'_i}$ where $p'_i \in P'$ for all $1\leq i\leq n$. The man and dog have identical walks on each straight edge $\segment{wp'_i}$. For the zigzag edges, the dog stays at the center $c$ of the bottleneck path and the man traverses all zigzag edges and arrives at the center. Then he continues his walk along with the dog. Thus, $T'$ has $n$ links so far.
	
	(ii) Box $B$ and $B'$: For the rest of the simplification, $T'$ selects the edges $\segment{s'_is_i}$, for all $1\leq i\leq m$.  In this case, the man and dog have similar walks on the path starting at $s'_i$ and ending at $s_i$. The dog stays on $s_i$ and the man walks through the zigzag edges and ends the walk on $s_i$ as well. However, for each path starting at $w$ and ending at $p$, where $p \in D_i$, the man and dog both reach $w$ at the same time, then the dog returns to $s_i$ and the man goes to $p$ only through the straight edge while staying closest to the dog and not going through the zigzag edges this time (the man has gone through the zigzag edges already as part of his walk in box $B'$, first stage). Note that this walk is possible since the two paths with straight edges from $w$ to $s_i$ and from $w$ to $p$ are within distance at most 1 from each other. Therefore $T'$ has  $m$ more edges, thus overall $n+m \leq n+k$ edges.
	
	Now suppose that there exists a minimum-complexity simplified tree $T'$ with the number of edges $n+m$ with $m \leq k$. We show that there exists a dominating set $S=\{s_1, \cdots, s_m\}$ for $G$ with $m\leq k$. First note that $T'$ selects at least one edge per bottleneck path in box $B'$ by construction, thus at least $n$ edges in box $B'$ overall. The remaining argument is for the edges of $T'$ and the edges $\segment{pp'}\in T$ for all $p \in P$ and $p'\in P'$ in box $B$.  We should show that if $T'$ selects at most $m$ edges in box $B$ such that the traversal distance between the edges in $T'$ and $T$ in box $B$ is at most 1, then there is a dominating set $S=\{s_1, \cdots, s_m\}$ for $G$ with $m\leq k$. Suppose that $T'$ selects $\{a_1,\cdots,a_m\} \subseteq P$. Since the dog should walk on $a_i$, with $i=1,\cdots, m$, the only possible vertices of $T$ for the man to walk on is $v \in D_i$. This means that there exists a set of unit disks $\{D_1,\cdots, D_m\}$ that covers all vertices in $P$ which corresponds to a dominating set of size $m\leq k$ for $G$. This completes the proof.
\end{proof}
	\begin{remark}
	The construction proposed in the proof of Theorem~\ref{thm:tree-tree-traversal} may not apply to the case under graph distance from $T$ to $T'$. Let $\mu$ be the mapping realizing $\Gdistance(T,T')\leq 1$. Since we need one edge per bottleneck path we have $\mu(c)=c$, and $\mu(w)=w$. For the points in $P'$ we take $\mu(p')= s'_i$ for all $i = 1,\cdots, m$ and all $p'\in D_i$.  Under such a mapping $\mu$, we have $\Frechet(\segment{p'c}, \P)> 1$, since $\|c - w\|> 1$, and $\| s'_i - w \| > 1$,  where $\P$ is the path starting at $\mu(p') = s'_i$ and ending at $\mu(c)= c$. Therefore $\Gdistance(T,T') > 1$, and the construction fails.
\end{remark}
\section{Vertex-Restricted   Tree-Tree Simplification under the Graph Distance }\label{sec:FPT-tree-tree}

In this section, we give a fixed parameter polynomial-time algorithm for vertex-restricted minimum-edge tree-tree simplification under graph distance from input to output. throughout the section we assume that $T$ is a rooted tree. The key idea is to consider the free space diagram between all edges of the input tree and all edges of the complete graph $G$ induced by the vertices of $T$. We denote this free space diagram between $G$ and $T$ by $\FSD(G,T)$. Then the optimal solution is a subtree of $G$ using a minimum number of vertices in $G$ such that there is a reachable path between each edge of $T$ and a subpath in $G$ across the respective free space diagram. Here, $\Gdistance$ can be either the strong or the weak graph distance, so we describe our generic algorithm which works for both versions of the graph distance.
Note that for a given tree $T$ and a graph $G$ deciding whether $\Gdistance(T,G)\leq \delta$ takes polynomial time~\cite{abksw-19}. Since the graph here is the complete graph whose edges are shortcuts (straight-line segments) between every pair of vertices in $T$, it is not hard to see that there is always a tree $T'\subset G$ for which $\Gdistance(T,T')\leq \delta$. This raises the question of whether one can compute a minimum-edge $T'\subset G$ such that $\Gdistance(T,T')\leq \delta$ or not.

We first compute the free space between every edge in $T$ and the entire graph $G$. Such a free space is called chunk. We then connect different chunks together with respect to the adjacency of edges in $T$. We then find a reachable path throughout $\FSD(G,T)$ that crosses the minimum number of vertices in $G$.
Given a tree $T = (V,E)$ in $\Reals^d$, we define the \emph{shortcut graph} of $T$ as $G=G(T)=(V(G),E(G))$, where $V(G)= V(T)$ and 
$E(G)=\lbrace \segment{uv}~|~u,v \in V_G \rbrace$. Each edge $e=\segment{uv} \in E(G)$ is linearly parameterized. 
The parameter space of $G$ is $E(G)\times [0,1]$ and the parameter space of $T$ is $E(T) \times [0,1]$. 
%
%
%
%
%

Now, let $\delta>0$, and consider the joint parameter space $E(T) \times [0,1] \times E(G)\times [0,1]$ of $T$ and $G$. Any $(e,\lambda,e',\gamma)\in E(T) \times [0,1] \times E(G)\times [0,1]$ is called \emph{free}
if $\|e(\lambda)-e'(\gamma)\|\leq \delta$, and the union of all free points are
referred to as the \emph{free space}.
%
A chunk is comprised of two main components: (1) spine; the free space between an edge in $T$ and a vertex in $G$, and (2) slice; the free space between a vertex in $T$ and the entire graph $G$; see \figref{fig:fss}. 
For any $v \in V(T), {v'}\in V(G)$, $e\in E(T)$ and $e' \in E(G)$ we call $\mathemph{\Spine(v')}= e \times [0,1] \times v'$ a \emph{spine}, and $\mathemph{\Slice(v)}= \cup_{e' \in E(G)} v \times e'\times [0,1]$ a \emph{slice}. We also denote the free space within a spine, and a slice as: $$\spine(v')=\lbrace (\lambda,v') ~|~ 0\leq \lambda \leq 1,\; \|e(\lambda)-{v'}\|\leq \delta \rbrace,$$ and: $$\slice(v)=\lbrace (v,\gamma) ~|~ 0\leq \gamma \leq 1,\; \|v-e'(\gamma)\|\leq \delta\rbrace,$$ respectively.
%
%
For an edge $e=\segment{uv}\in E(T)$, it holds that  $\Slice(u), \Slice(v)\subset \FSD(e,G)$ and $\Slice(u)$ is a subset of all free spaces with respect to edges in $T$ incident on $u$.

\begin{definition}
	For every $v \in V(T)$, 
	a free space interval $I \in \slice(v)$ is called \emph{elementary} if $I$ lies completely within $v \times e'\times[0,1]$, for some $e' \in E(G)$; see Figure~\ref{fig:fss}. 
\end{definition}
\begin{figure}[!h]
	\centering
	\includegraphics[width=.8\textwidth]{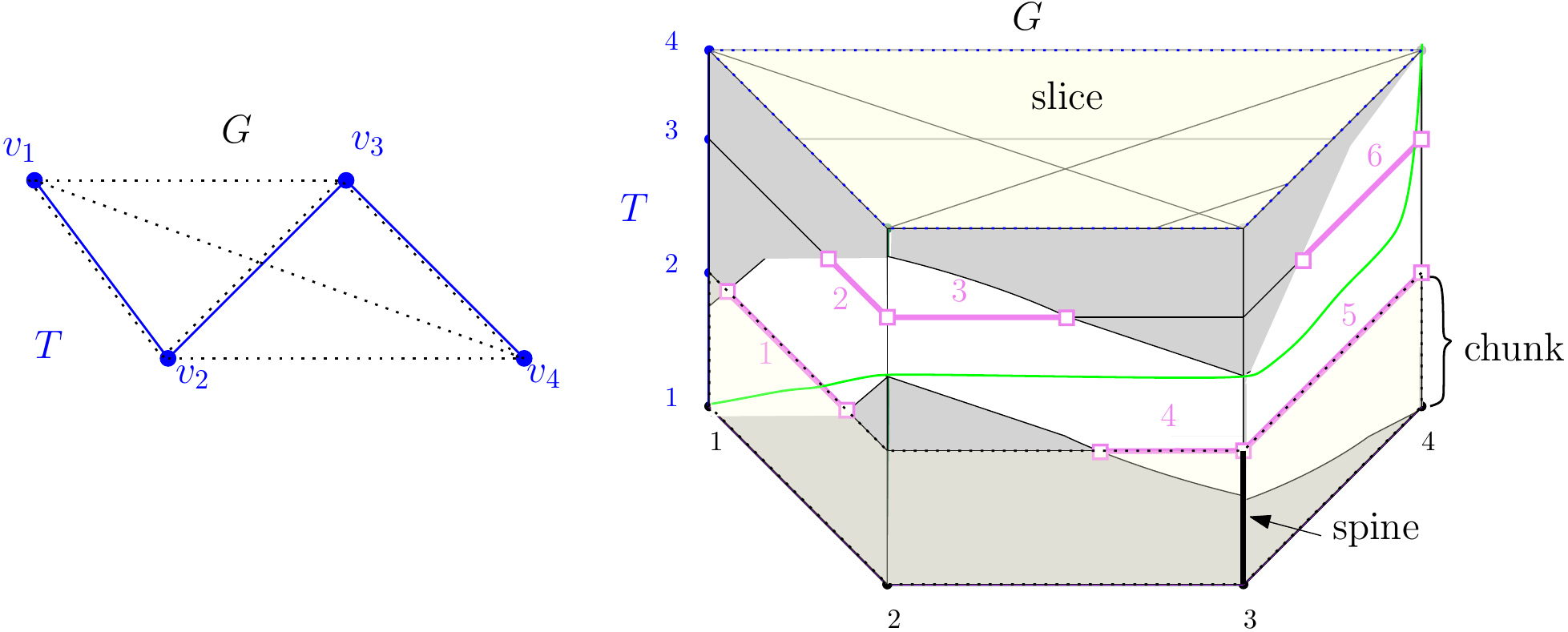}
	\caption[Elementary intervals]{An example of a tree, its complete graph in dotted line, and the free space surface between them. There are six elementary intervals in the free space surface. A reachable path highlighted in green uses a minimum number of spines (vertices of $G$) while traversing all the edges in $T$.}
	\label{fig:fss}
\end{figure}

Our aim is to  propose a dynamic programming algorithm for  this problem. Before we get to the algorithm we need to introduce some notions. Let $T$ be a rooted tree. In this setup every vertex $v \in V(T)$ is the parent of a set of vertices if (1)  $v$ is the neighbor of the vertices, and (2) takes fewer vertices on $T$ to reach the root of $T$. The main idea of our algorithm is to start constructing optimal trees rooted at elementary intervals of $\slice(v)$ for all $v\in V(T)$, where $v$ is a leaf and then propagate the minimum-link rooted subtrees bottom-up towards the root in a dynamic programming fashion. For this, we associate a cost function $\psi:[0,1]\rightarrow \Nats $ with each elementary interval $I$, where $\psi(I)$ is the number of vertices in a minimum-vertex simplified subtree rooted at $I$. When a simplified tree $T'$ \emph{is rooted at an elementary interval} $I \in \slice(u)$ it means that $T'$ simplifies the subtree of $T$ rooted at $u$, and the mapping realizing $\Gdistance(T,T')\leq \delta$ matches $u$ to a point $x \in I$. Let $T'_I$ denote a simplified tree rooted at $I$, and let $\weight(T'_I)$ be the \emph{weight} of  $T'_I$ that indicates the number of vertices on it. Clearly, for an $I \in \slice(u)$, $\psi(I)= \min_{T'_I} \weight(T'_I)$.

We also extend our notations as follows: Suppose in a rooted tree $T$, $u$ is the parent of $v$ in $T$ and $\Cld(u)= \lbrace {v_1},\cdots,{v_k}\rbrace$ is the set of children of $u$ and correspondingly $\Prnt(v_i)=u$, for all $i=1,\cdots,k$. 
Each $\slice(v_i)$ consists of a set of elementary intervals, for all $i\in \{1,\cdots, k\}$. Note that an elementary interval might belong to different sets of elementary intervals associated with different slices. In other words, each $\slice(v_i)$, for all $i \in \{1,\cdots,k\}$, has a set of elementary intervals as depicted in Figure~\ref{fig:fss} and an elementary interval $I$ can belong to $\{\slice(v_{i_1}),\cdots, \slice(v_{i_m})\}$ where $\{i_1,\cdots, i_m\} \subseteq \{1,\cdots, k\}$.  At this point, we say that $I \in \slice(v_j)$ \emph{covers} a subset of vertices $\{v_{i_1},\cdots, v_{i_m}\}$ where $\{i_1,\cdots, i_m\} \subseteq \{1,\cdots, k\},$ and we denote this by $C(I)=\{v_{i_1},\cdots v_{i_m}\}$ for some $j \in \{1,\cdots,k\}$.  Suppose $\I^*= \{I_1,\cdots, I_\ell\}$ is the universal set of all elementary intervals over all $\slice(v_i)$ for all $i \in \{1,\cdots, k\}$. 
For any $I\in \slice(u)$ where $u$ is a not a leaf in $V(T)$, we propose the following recursive formula:

\begin{align}
\psi(I) &= \min\limits_{\I \subseteq \I^*} \sum\limits_{I' \in \I} \big(\psi(I') + \kappa(I',I)\big), \label{formula:DP_tree}
\end{align}

where $\I$ ranges over all subsets of $\I^*$ such that $\cup_{I\in \I} C(I) = \Cld(u)$, i.e., 
all elementary intervals in $\I$ together cover all vertices in $\Cld(u)$. Here, $\kappa(I,I')$ is the number of vertices in a minimum-edge reachable path between $I \in \slice(u)$ and $I'\in \slice(v_i)$ in $\FSD(e,G)$ where $e=\segment{uv_i}$. In other words, $\kappa(I,I')$ is the minimum number of spines (vertices in $G$) that a reachable path has to cross to reach $I'$ starting at some point in $I$. If $u$ is a leaf then $\psi(I)= 1$.

We are now ready to prove the correctness of the formula:

\begin{lemma} [Correctness] \label{lem:correctness}
	Let $I\in \slice(u)$ for some $u \in V(T)$. Then the recursive formula (\ref{formula:DP_tree}) correctly computes $\psi(I)$. 
\end{lemma}

\begin{proof}
	We prove by induction. First we consider the case where $u$ is a leaf in $T$ as the inductive base. Let $\mu$ be a mapping realizing $\Gdistance(T,T')\leq \delta$. For the sake of minimality of $T'_I$ there should only be one point $p \in I$ such that $\mu(u)=p$. Therefore $T'_I$ is a single-vertex tree and $1\leq \psi(I)\leq \weight(T'_I) = 1$.
	
	Now we consider the case where $u$ is not a leaf in $T$. Suppose that $T^*_I$ is an optimal simplified tree rooted at $I$. Observe that $T^*_I$ passes through a set of elementary intervals $\I\subseteq \I^*$ covering all vertices in $\Cld(u)$.  
	Therefore the number of vertices of $T^*_I$ is equal to the sum of the weights of subtrees rooted at each interval $I'$ in $\I$ and the weights of simple paths $P$ connecting each of those intervals to $I$: 
	$$\psi(I)= \min\limits_{T'_I} \weight(T'_I) = \weight(T^*_I) = \sum\limits_{I' \in \I} (\weight(T^*_{I'})+\weight(P(I',I))).$$
	
	Note that $\weight(P(I',I)) = \kappa(I',I)$ otherwise $T^*_I$ would no longer be optimal. Also we know that $\weight(T^*_{I'}) = \psi(I')$ by the inductive hypothesis.  
	Thus we have: 
	
	$$\psi(I)= \weight(T^*_I) = \sum\limits_{I' \in \I} \big(\psi({I'})+\kappa(I',I)\big).$$
	
	Realize that $\I$ is a subset of $\I^*$ for which the expression $\sum_{I' \in \I} \big(\psi({I'})+\kappa(I',I)\big)$ is minimum compared to other subsets of $\I^*$ due to the optimality of $T^*_I$. Therefore: 
	
	$$\psi(I)= \weight(T^*_I) = \min\limits_{\I \subseteq \I^*} \sum\limits_{I' \in \I} \big(\psi({I'})+\kappa(I',I)\big),$$ as desired.
\end{proof}

We now present our dynamic programming algorithm in more detail as follows: Similar to~\cite{abksw-19}, we first compute the connected components of all vertices in $T$ (the part of $G$ lying within the balls of radius $\delta$ around the vertices of $T$). For two neighboring vertices $u$ and $v$, we prune all invalid paths between $u$ and $v$ in $G$ for which the Fr\'echet distance to $e=\segment{uv}$ is greater than $\delta$. This gives us a pruned graph $G'$. Now we construct $\FSD(G',T)$ and define the slices with respect to $G'$ and $T$.
We perform a BFS search on $T$ starting at the root and store the vertices into an auxiliary stack $S$ in the order they are being encountered along the search. If a vertex $u$ is a leaf in $T$ then we set $\psi(I)=1$ for all $I\in \slice(u)$ and $\psi(I)=\infty$, otherwise. Once we processed all the vertices in $T$, we pop the vertices from $S$. For every popped vertex $u$ and every elementary interval $I\in \slice(u)$, we compute $\psi(I)$ using the recursive formula. We repeat the process until we reach the root $r$ and $I$, for all $I \in \slice(r)$. In the end, we backtrack and find those elementary intervals in $\FSD(G',T)$ that contain the minimum $\psi$ values. The vertices of the simplified tree are the endpoints of those intervals that end (or start) at their neighboring spines.   

\begin{lemma} \label{lem:G'}
	Constructing $G'$ takes $O(n^3)$ under the weak graph distance and $O(n^5)$ under the strong graph distance.
\end{lemma}

\begin{proof}
	According to Lemma 4 in \cite{abksw-19}, pruning all invalid paths in $G$ takes $O\Big(|E(T)|\cdot|E(G)|\Big)$ under the weak graph distance and $O\Big(|E(T)|\cdot|E(G)|^2\Big)$ under the strong graph distance. Since $|E(T)|=|V(T)|- 1 = n-1$, and $|E(G)|= |V(T)|^2= n^2$, therefore the upper bounds can be obtained.
\end{proof}

\begin{lemma} \label{lem:kappa}
	For every $I \in \slice(u)$ and $I' \in \slice(v)$ where $v \in \Cld(u)$, there exists a procedure that computes $\kappa(I,I')$ in $O(n^2)$ time under both the weak and strong graph distances.
\end{lemma}

\begin{proof}
	Given two start and end points $I$ and $I'$, $\kappa(I,I')$ can be computed by finding the shortest path in $\FSD(G',T)$ between $I$ and $I'$ that crosses the minimum number of spines, if it exists. If such a path does not exist then we set $\kappa(I,I')=\infty$. 
	This corresponds to computing the shortest path in $G'$ starting at a point in $I'$ and ending at some point in $I$. Computing the shortest path only takes $O\Big(|V({G'})|+|E({G'})|\Big)= O(n^2)$ using a BFS search in $G'$ for both distances. Therefore, the total runtime under both weak and strong graph distances is $O(n^2)$.
	
	
\end{proof}

\begin{theorem} [Runtime] \label{thm:fixedTTGD}
	Let $T$ be a tree in $\Reals^d$ with $n$ vertices and $\delta>0$. There exists an algorithm that computes a  vertex-restricted min-complexity simplified tree $T'$ in $O(\alpha^2 2^\alpha n^3)$ and $O(n^5+\alpha^2 2^\alpha n^3)$ times fulfilling $\WGdistance(T,T')\leq \delta$ and $\Gdistance(T,T')\leq \delta$, respectively, where $\alpha$ is the maximum number of elementary intervals over all slices in $\FSD(G',T)$.
\end{theorem}
\begin{proof}
	Let $\T(n)$ be the runtime for constructing $G'$. Depending on the type of the weak or strong graph distance, $\T(n)$ can be different according to Lemma~\ref{lem:G'}. At the beginning part of our dynamic program, the BFS search on the vertices of $T$ together with the stack operations takes $O(n)$ time overall. Now it only remains to show the runtime of computing the recursive formula $\psi(I)$ per $I\in \slice(u)$ when $u$ is not a leaf in $T$. Recall that $\I^*= \{I_1,\cdots, I_\ell\}$ is the set of all elementary intervals over all $\slice(v_i)$ for all $i \in \{1,\cdots, k\}$.
	Given a universal set $\Cld(u)=\{v_1,\cdots,v_k\}$, computing all sets $\I\subseteq \I^*$ covering the entire universal set $\Cld(u)$ together with minimum cost of $\sum_{I' \in \I} (\psi(I')+\kappa(I',I))$,  where every $I \in \I$ consists of a subset of $\Cld(u)$, is equivalent to solving the \textsc{Weighted  Set-Cover} problem.  A brute-force algorithm takes $O(k2^{\ell})$ to compute all sets $\I\subseteq \I^*$ covering the entire universal set $\Cld(u)$ per $I\in \slice(u)$. Note that the runtime of the procedure computing $\kappa(I,I')$ is already $O(n^2)$ (Lemma~\ref{lem:kappa}). Thus computing the formula per $I$ takes $O(k2^{\ell}n^2)$ so far. Having $\ell$ intervals like $I$ in $\slice(u)$ and $n$ vertices like $u$ to process, yields the total runtime of $O(\T(n)+ k\ell2^\ell n^3)$. 
	
	Now suppose the number of intersections between the edge set of $G'$ and ball of radius $\delta$ around $u$ is at most $\alpha$, i.e., the maximum number of elementary intervals over all slices of $\FSD(G',T)$. Realize that since there is always an optimal  solution to the problem, so there is at least one elementary interval on every slice in $\FSD(G',T)$ and thus $\alpha\geq 1$. On the other side, $\alpha\leq  |E(G')|\leq |E(G)| = \frac{n(n-1)}{2}$, by definition, therefore $\alpha \in [1,\frac{n(n-1)}{2}]$. Note that the number of children of $u$ is also at most $\alpha$, i.e., $k\leq \alpha$. On the other hand, the number of elementary intervals $I\in \slice(u)$ is at most $\alpha$ hence $\ell \leq \alpha$ as well. Therefore in this case, the total runtime is $O(\T(n)+\alpha^2 2^\alpha n^3)$. Overall, the total runtime under the weak graph distance would be: 
	
	$$T(n)= O\big((1+\alpha^2 2^\alpha) n^3\big) = O\big(n^3\alpha^2 2^\alpha\big),$$ since $\alpha \geq 1$. The runtime under the strong graph distance would also be: 
	
	$$T(n) = O\big(n^5+\alpha^2 2^\alpha n^3\big).$$
	This completes the proof.
\end{proof}

\section{Subgraph-Restricted Graph-Graph Simplification under the Graph Distance}\label{sec:NP-vertex}

In this section we prove that computing the subgraph-restricted minimum-complexity simplification under $\Gdistance(G,G')\leq \delta$ is NP-hard. 
Throughout the section we assume that the simplified graph is connected as well as the input graph.
We reduce from a specific variant of {\textsc{Max-2SAT}}
problem which is defined as follows: Given a set of variables $X=\{X_1,X_2,\cdots,X_n\}$ and a CNF-SAT formula $\F$ consisting of $m$ disjunctive clauses each with at most two variables (including the negation), find an assignment $\{\true,\false\}^n$ to the variables in $X$ such that the number of satisfied clauses is maximum. The decision version of this problem, $\textsc{Max-2SAT}(X,\F,k)$, takes an integer $k>0$ as an argument and asks whether the exists an assignment $A$ to $X$ under which the number of satisfied clauses in $\F$ is at least $k$ or not. 

\begin{definition}[Variable Graph]
	Given a formula $\F$, a variable graph $G_\F$ is a graph whose nodes are the variables in $\F$. Two nodes are connected by an edge if the respective variables of the nodes belong to the same clause in $\F$.
\end{definition}

A \textsc{Max-2SAT} problem is said to be \emph{\textsc{Bipartite-Max-2SAT}} if the variable graph of its formula forms a bipartite graph.  We first prove that \textsc{Bipartite-Max-2SAT} is NP-hard and then we make our main reduction from \textsc{Bipartite-Max-2SAT} to prove that the problem of interest is NP-hard as well.
%

%
\begin{lemma}
	\textsc{Bipartite-Max-2SAT} is NP-hard.
\end{lemma}
\begin{proof}
	We prove this by a reduction from \textsc{3-SAT}. Consider an instance of \textsc{3-SAT}, which consists of a CNF formula $\F$ in which all clauses have at most 3 variables.
	We will construct another formula
	for \textsc{Bipartite-Max-2SAT} whose clauses have at most two variables
	and its variable graph is a bipartite graph (see Figure~\ref{fig:treemax2sat}).
	Let $C=(x \vee y \vee z)$ be an arbitrary clause in $\F$.
	We replace $C$
	with 16 clauses yielding the formula $\F_C$ as follows:
	
	\begin{figure}
		\centering
		\includegraphics[width = 0.5\textheight]{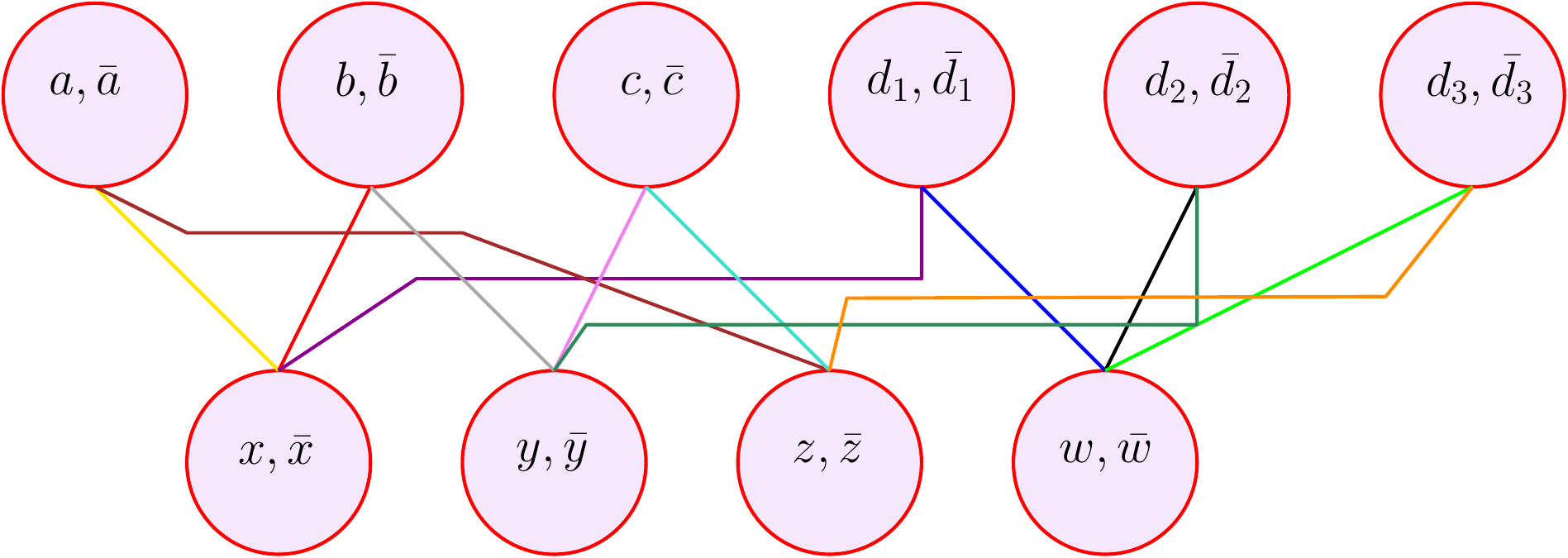}
		\caption[Bipartite variable graph]{The variable graph of the new clause $\F_C$ forms a bipartite graph.}
		
		\label{fig:treemax2sat}
	\end{figure}
	\begin{align*}
	\F_C = &(\bar{x}\vee a) \wedge (\bar{x}\vee \bar{b}) \wedge(\bar{y}\vee b) \wedge(\bar{y}\vee \bar{c}) \wedge(\bar{z}\vee c) \wedge(\bar{z}\vee \bar{a}) \\
	&\wedge(\bar{w}\vee d_1) \wedge(\bar{w}\vee d_2) \wedge(\bar{w}\vee d_3) \\
	&\wedge(x\vee \bar{d_1})  \wedge(y\vee \bar{d_2}) \wedge(z\vee \bar{d_3})   \\ 
	& \wedge(x) \wedge(y) \wedge(z) \wedge(w), \notag
	\end{align*}
	where $a,b,c,w,d_1,d_2,d_3$ are additional variables, and $k=13m$.
	The transformed instance $\F'$ is the conjunction of all $\F_C$ for all clauses $C$ in $\F$.
	%
	%
	We have the following cases:
	
	\begin{itemize}
		
		\item If $x=y=z=\true$, we know that at most 3 of the clauses $(\bar{x}\vee a) \wedge (\bar{x}\vee \bar{b}) \wedge(\bar{y}\vee b) \wedge(\bar{y}\vee \bar{c}) \wedge(\bar{z}\vee c) \wedge(\bar{z}\vee \bar{a})$ can be satisfied, no matter which values are assigned to $a,b,c$. By setting $w=d_1=d_2=d_3= \true$ we have at most 13 clauses of $\F_C$ satisfied. 
		
		\item If $x=y=\true$, $z=\false$, by setting $a=\true$, $b=\true$ and $c=\false$ we have 5 of the clauses $(\bar{x}\vee a) \wedge (\bar{x}\vee \bar{b}) \wedge(\bar{y}\vee b) \wedge(\bar{y}\vee \bar{c}) \wedge(\bar{z}\vee c) \wedge(\bar{z}\vee \bar{a})$ satisfied. Now if we set $d_1 = d_2 = w = \false$ we will have 7 more clauses satisfied. Choosing either $d_3 =\true$ or $d_3=\false$ only yields one more satisfied clause, therefore we have at most 13 clauses in $\F_C$ satisfied. 
		
		\item If $x=\true$, $y=z=\false$, by setting $a=\true$ and $b=\false$ we have all the 6 clauses $(\bar{x}\vee a) \wedge (\bar{x}\vee \bar{b}) \wedge(\bar{y}\vee b) \wedge(\bar{y}\vee \bar{c}) \wedge(\bar{z}\vee c) \wedge(\bar{z}\vee \bar{a})$ satisfied. Now if we set $d_1 = w = \false$ we will have 5 more clauses satisfied. Choosing any values for $d_2$ and $d_3$ only adds two more satisfied clauses. Therefore we have at most 13 clauses satisfied.
		
		\item If $x=y=z = \false$, regardless setting any values to $a,b,c$, we have always all the 6 clauses $(\bar{x}\vee a) \wedge (\bar{x}\vee \bar{b}) \wedge(\bar{y}\vee b) \wedge(\bar{y}\vee \bar{c}) \wedge(\bar{z}\vee c) \wedge(\bar{z}\vee \bar{a})$ satisfied. Now setting $w = d_1 = d_2 = d_3 = \false$ we have at most 6 other clauses satisfied. Therefore, at most 12 clauses of $\F_C$ are satisfied.  
	\end{itemize}
	
	Now suppose there exists an assignment satisfying all $m$ clauses in $\F$. We have to show that $\F'$ has at least $k=13m$ clauses satisfied. If $m$ clauses in $\F$ are satisfied, since each clause $C$ in $\F$ is substituted with 16 clauses $\F_C$ and 13 of them are satisfied (since $C$ is satisfied according to above cases), then $k=13m$ of them are satisfied overall.
	
	Now for the other direction, assume that $m'_s$ clauses in $\F'$ are satisfied and $m'_s\geq k=13m$. Let $m_s$ be the number of clauses satisfied in $\F$.
	%
	For the sake of contradiction, assume that $m_s<m$. Then we have
	$13m\leq m'_s = 13m_s+12(m-m_s) = m_s+12m<13m$ which is a contradiction.
\end{proof}

We are now ready to present our main reduction from \textsc{Bipartite-Max-2SAT} to the minimum-complexity graph-graph simplification problem. 

\paragraph{The reduction:} Set $\delta=1 $ and construct
a graph $G$ from $G_\F$ as follows: $G$ consists of two types of gadgets;
\emph{variable} and \emph{clause} gadgets. A variable gadget $X_i$ has
two vertices $\true_i$ and $\false_i$ representing the two possible
assignments to $X_i$. 

A binary clause gadget of $(X_i\vee X_j)$ connects two
pairs of vertices of two variable gadgets using paths passing through some \emph{hook} vertex $h$  that is located far away from the variables. The hook vertex controls the number of links such that a clause gadget can be simplified using exactly two edges if its corresponding clause is satisfied, and three edges otherwise, under the corresponding assignment to the variables in the clause. There are four interior vertices connected to the hook that are located close enough, by distance $\epsilon$, to each other, where $0<\epsilon \ll \delta = 1$. The \emph{length} of the clause gadget is $L_{ij}$ that is the difference between the $y$-coordinates of the two variable gadgets. Here we also set $L_{ij}\gg 1$.
See Figure~\ref{fig:cgadget} for further illustration on two types of clause gadgets.

\begin{figure}[h]
	\centering
	\includegraphics[width=0.50\textheight]{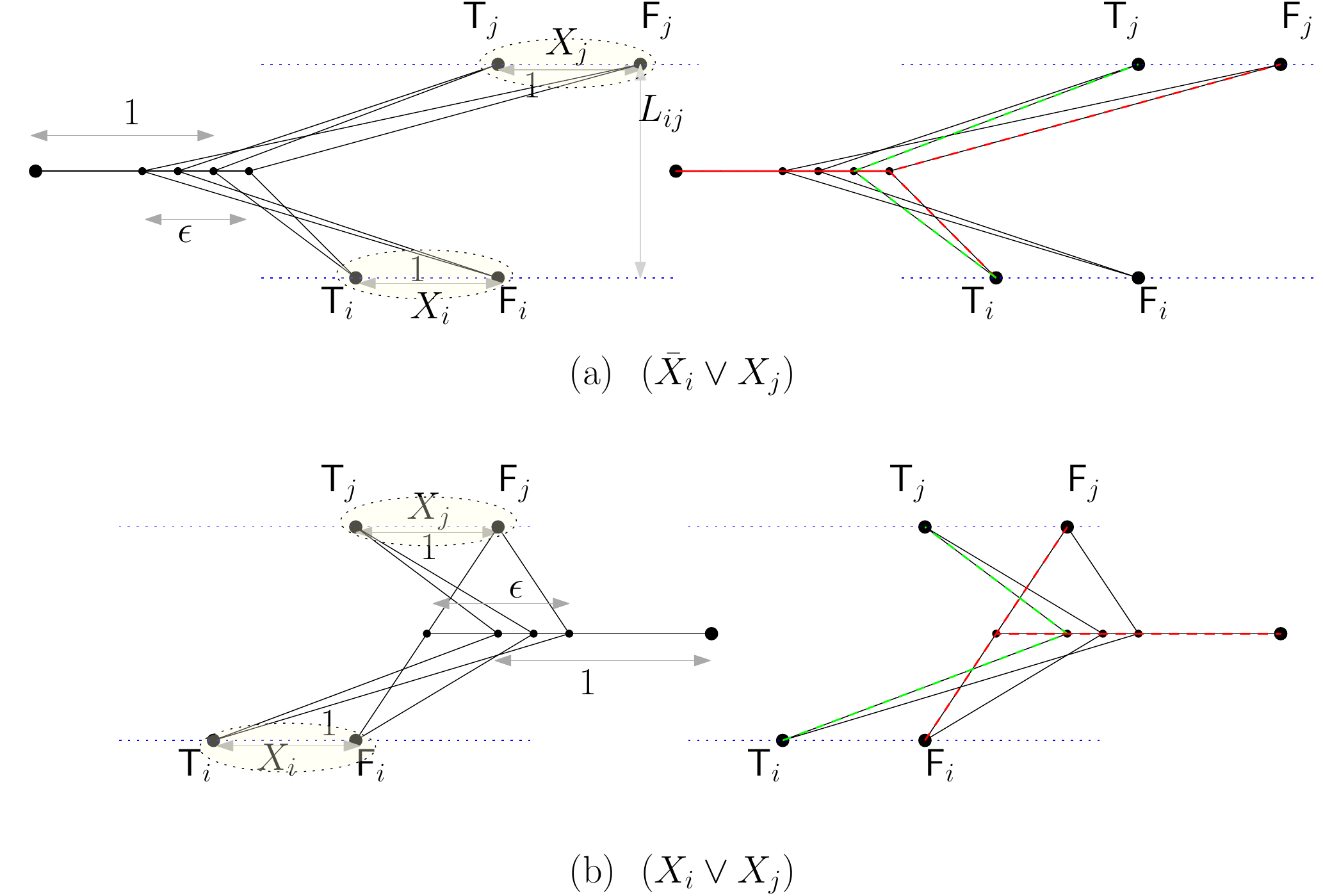}
	\caption[Clause gadgets]{The clause gadgets. (a) is a clause gadget for $(\bar{X}_i\vee X_j)$. Note that a 2-edge path, in green, from $T_i$ to $T_j$ simplifies the gadget while satisfying the corresponding clause. A tree with 3 edges from $\true_i$ to $\false_j$, in red, which corresponds to an assignment that does not satisfy the clause. (b) is a clause gadget for $(X_i\vee X_j)$, thus using the 2-link path from $\false_i$ to $\false_j$ requires and additional link. Here, $\eps>0$ is a sufficiently small real number. The constructions for $({X}_i\vee \bar{X_j})$ and $(\bar{X_i}\vee \bar{X_j})$ are symmetric to case (a) and case (b), respectively.}
	\label{fig:cgadget}
\end{figure}

The clause gadget for every unary clause ($X_i$) is similar to case (b) and involves one additional variable gadget $X_t$. To handle this, we remove the path between $\true_t$ and $\false_i$, as well as the path between $\false_t$ and $\true_i$, from the original clause gadget in case (b). 

Let $G_\F=(\mathcal{X}_1\dot{\cup}\mathcal{X}_2, E)$, so $\mathcal{X}_1\dot{\cup}\mathcal{X}_2$ is a partition of the variables, $E$ be the set of edges/clause gadgets, and $G_\F$ is a bipartite graph. 
Let $\ell_1$ and $\ell_2$ be two lines parallel to the $x$-axis, at vertical distance $L$ from each other. Note that $L = L_{ij}$.
We place the variable gadgets belonging to $\mathcal{X}_1$ on $\ell_1$ and those
belonging to $\mathcal{X}_2$ on $\ell_2$. Since $G_\F$ is bipartite, each clause gadget connects a variable from $\ell_1$ to another variable from $\ell_2$ (see Figure~\ref{fig:biembedding}).

Now, we need to make sure that the edges (links) on the subgraph simplifying a clause gadget can be used only to \emph{cover} the edges of that clause. By ``cover'', we mean that the graph distance from the clause gadget to the subgraph simplifying the clause gadget must be 1. 

For this, we choose $L\gg 1$. The variable gadgets are spaced apart at distance larger than $1$ along $\ell_1$ and $\ell_2$. This way the clause gadgets of the same type (either type (a) or type (b)) do not overlay on top of each other and cannot \emph{cover} for each other. The following lemmas lead to our main theorem in this section. In particular Lemma~\ref{lem:onevertex} below demonstrates that there is a consistency between the literal ($\{
\true,\false\}$) assigned to each variable and the corresponding vertex the simplification selects in the variable gadget.
%

\begin{figure}
	\centering
	\includegraphics[width=0.35\textheight]{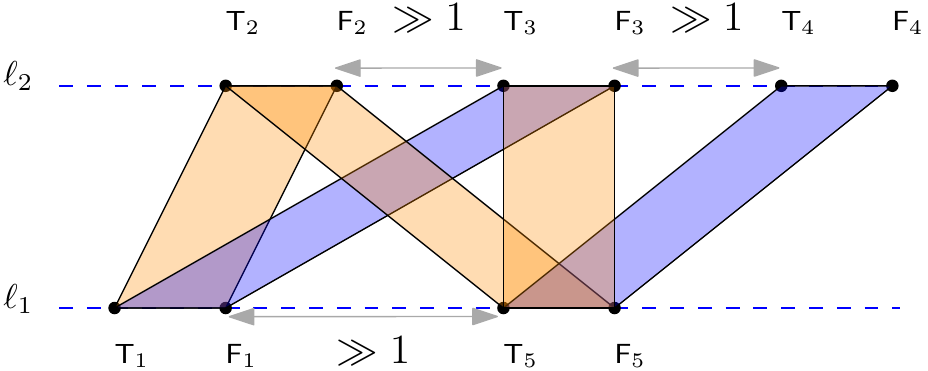}
	\caption[The embedded bipartite variable graph]{A schematic representation of the thickened embedded variable graph induced by $(X_1 \vee X_2) \wedge (X_1 \vee \bar{X}_3) \wedge (\bar{X}_2 \vee \bar{X_5}) \wedge (X_4 \vee \bar{X}_5) \wedge (X_3 \vee X_5)$. For a better presentation of this embedding, the type (a) and type (b) clause gadgets are highlighted in orange and blue strips, respectively.
	}
	\label{fig:biembedding}
\end{figure}

\begin{lemma} \label{lem:onevertex}
	A min-complexity simplified graph selects exactly one vertex per variable gadget.
\end{lemma}  
\begin{proof}
	We use a proof by contradiction. Suppose $G'$ is a minimum-complexity simplified graph that does not choose exactly one vertex per variable gadget. Let $X_i$ be a variable in $\F$ where $G'$ chooses no vertex from its gadget in $G_\F$. By construction it immediately follows that $\Gdistance(G,G')>1$ since the vertices of the variable gadget $X_i$ (either $\true_i$ or $\false_i$) cannot be mapped to anywhere in $G'$ with $L\gg 1$.  Now suppose $G'$ chooses two vertices from $X_i$ and w.l.o.g. let $X_j$ and $X_k$ be two
	variable gadgets adjacent to $X_i$, and $\mu$ is a graph mapping realizing $\Gdistance(G,G')\leq1$ (see Figure~\ref{fig:onevertex}).  
	To distinguish between the vertices of $G'$ and $G$ in the case that they lie on each other, we give their vertices different names to identify which vertex belongs to which graph. Note that since this is a vertex-restricted case, thus $V(G') \subset V(G)$. Now let $t_i, f_i \in V(G')$ be two vertices of $G'$, where $t_i= \true_i$ and
	$f_i= \false_i$. Note that $t_i$ cannot be connected to $f_i$ by some edge because it is against the subgraph-restrictedness.
	Correspondingly,
	$G'$ selects two paths $\P_{ij}$ and $\P_{ik}$ connecting $\true_i$ to
	some vertex of $X_j$ and $\false_i$ to some vertex of $X_k$,
	respectively.  Since $G'$ selects two different vertices of $X_i$ variable gadgets, $\P_{ij}$ and $\P_{ik}$ are not connected at $X_i$ gadget. This means that there must be a variable gadget $X_w$ such that $\P_{ij}$ and $\P_{ik}$ are connected to each other through $X_w$ because we want the simplified graph be connected as well as the input graph. W.l.o.g. suppose $\P_{ij}$ is the one that is connected to a vertex $\true_w$ of $X_w$ and is placed to the right of $X_i$.
	Thus, there are two following cases:
	\begin{itemize}
		\item  $\mu(\true_i) = t_i$. Now consider the edge $\segment{\true_ib} \in E(G)$. Clearly, $\Frechet\big(\segment{\true_ib}, \P_{ij}[t_i, \mu(b)]\big)>1$. This is because the distance between $a$ and any point on $\segment{\true_ib}$ is larger than 1 by construction. Therefore, $\Gdistance(G,G')>1$.
		
		\item $\mu(\true_i)= f_i$. This time consider the edge $\segment{\false_ia} \in E(G)$. Clearly, $\Frechet\big(\segment{\false_ia}, P_{ik}[f_i, \mu(a)]\big)>1$. This is because the distance between $b$ and any point on $\segment{\false_ia}$ is larger than 1 by construction. Therefore, $\Gdistance(G,G')>1$.
		
	\end{itemize} 
	Thus, we have a contradiction in both cases above and $G'$ selects exactly one vertex of $X_i$.
\end{proof}

\begin{remark}
	The proof above may not be credible for the case under the traversal distance from $G$ to $G'$, since both $t_i$ and $f_i$ can be chosen and still there can be a traversal on $G$ and $G'$ under which $\Traversal(G,G')\leq 1$.  	
\end{remark}

\begin{figure} 
	\centering
	\includegraphics[width=0.35\textheight]{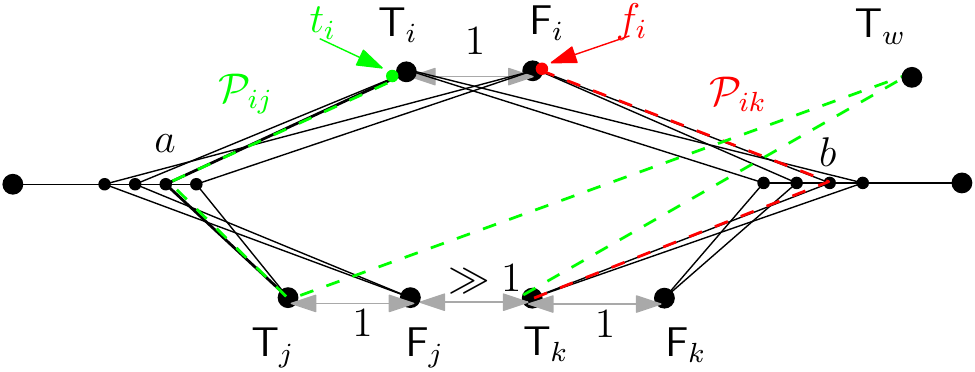}
	\caption[Proof of the consistency of the construction]{Choosing two vertices $\true_i$ and $\false_i$ by $G'$ results in having paths $\P_{ij}$ and $\P_{ik}$, highlighted in green and red dashed lines, respectively, whose graph distances are larger than 1 from their respective clauses.}
	\label{fig:onevertex}
\end{figure}

\begin{lemma} 
	
	Let $m$ be the number of clauses. An assignment satisfying at least $k<m$ clauses of \textsc{Bipartite-Max-2SAT} exists if and only if $G'$ with $|E(G')|+|V(G')|\leq 6m-2k+1$ exists.
	
\end{lemma}

\begin{proof}
	We first argue on the number of edges and then complete our proof on the total number of edges and vertices. We show  that an assignment satisfying at least $k$ clauses exists if and only if a  simplified graph $G'$ with $|E(G')|\leq 3m -k$ exists. Let $A$ be an assignment under which at least $k$ clauses of $\F$ are
	satisfied. Now let $G'$ be a minimum-edge simplified graph whose number of edges is $|E(G')|$. Observe that since $A$ assigns a literal to $X_i$, $G'$ can select exactly one
	vertex ($\true_i$ or $\false_i$) that is assigned to $X_i$ under
	$A$. Let $m_s$ be the number of satisfied clauses. Since $m_s\geq k $,
	by construction we have:
	
	$$ |E(G')| = 2m_s+3(m-m_s) = 3m-m_s \leq 3m-k,$$
as desired. 	
	Now for the other direction, let $G'$ be an optimal simplified graph whose number of edges is
	$|E(G')| \leq 3m-k$. Following Lemma~\ref{lem:onevertex}, $G'$ only chooses
	one vertex per variable gadget. This implies the vertices of $G'$
	obtain a set of assignments $A$ to variable set $X$ of $\F$ in
	the \textsc{Bipartite-Max-2SAT} problem. Observe that $G'$ takes
	either two-edge or three-edge tree per clause gadget.  Let $S_2$ and
	$S_3$ be the set of all such two-edge and three-edge trees over all
	clause gadgets, respectively. Let $m_s$ and $m_u$ be the number of
	satisfied and unsatisfied clauses in $\F$, respectively. In other words, $|S_2|=m_s$
	since every two-edge tree in $S_2$ is satisfying the respective clause 
	and $|S_3|=m_u$ since every three-edge tree in $S_3$
	is unsatisfying the respective clause. Now we have: $|E(G')| =
	2|S_2|+3|S_3|$, and $m=m_s+m_u$. We can obtain that $m_u=m-m_s$.  On
	the other hand $|S_2|=m_s$ and $|S_3|=m_u$, thus it follows that: $|E(G')|=
	2m_s+3m_u$. Thus we have $m_u = (|E(G')| -2m_s)/3$.  Now we can set:
	
	$$m-m_s=(|E(G')|-2m_s)/3 \implies 3m-3m_s \leq 3m-k -2m_s \implies m_s \geq k. $$	
	Therefore, $A$ satisfies at least $k$ clauses of $\F$.
	Realize that in our construction each clause is simplified by a tree, hence the number of edges and vertices of $G'$ differs only by 1  per clause gadget. This implies that $|V(G')| = |E(G')|+1 = 3m-k+1$. Thus, $|E(G')|+|V(G')|\leq 3m-k+3m-k+1 = 6m-2k+1$ if and only if $A$ satisfies at least $k$ clauses of $\F$.
\end{proof}

\begin{theorem}  \label{thm:subgraph-graph-graph}
	The subgraph-restricted min-complexity graph-to-graph simplification under $\Gdistance(G,G')\leq \delta$ is NP-hard.
\end{theorem}

\section{Leaf-Restricted Tree-Tree Simplification from Input to Output}
We now consider the case where the leaves of $T'$ are identical to subset of the leaves of $T$. We show that given a tree $T$ in $\mathbb{R}^2$ and a threshold $\delta$, it is NP-hard to compute a tree $T'$ with a minimum complexity such that $\Gdistance(T,T')\le\delta$, where every leaf in $V(T')$ is identical to some leaf in $V(T)$. Recall that the construction proposed in Section~\ref{sec:vertexrestricted} works for the vertex-restricted case only under the traversal distance between the two trees. Similar to Theorem~\ref{thm:tree-tree-traversal} we reduce from the (\textsc{MDSUDG}) problem. The reduction is a special case of the construction  presented in Section~\ref{sec:vertexrestricted}. In addition, with a similar proof we can show that the problem under traversal distance is NP-Hard since unlike the proof (and construction) in Theorem~\ref{thm:tree-tree-traversal} we do not use a bottleneck path in our reduction anymore. Unlike that construction we do not use bottleneck paths but only straight-line edges directly connecting the leaves to $w$. The reduction is very straightforward as follows:
Given a unit disk graph $G=(V,E)$ with $P =V(G)$, let $\delta=1$ and $B$ be the smallest bounding box of unit disks around the points in $P$. Construct  $T$ as follows: 

$$V(T)=P\cup \{w\},~\mbox{and}~E(T)=\{\segment{pw}\mid p\in P\},$$ for some point $w$ which lies far enough from $B$.
We have the following theorem:


\begin{theorem} \label{thm:leaf-restricted-tree}
	Let $\delta>0$. Given a tree $T$ in $\Reals^2$, computing  a leaf-restricted min-complexity tree-tree simplification $T'$ under $\Gdistance(T,T')\leq \delta$ and $\Traversal(T,T')\leq \delta$ is NP-Hard.
\end{theorem}
\begin{proof}
	Let $k>0$ be an integer as the decision parameter to the decision version of \textsc{MDSUDG} whose instance is a unit disk graph $G=(V,E)$. We show that there exists a dominating set $S = \{s_1,\cdots, s_m\}$ with $m\leq k$ to the (\textsc{MDSUDG}) if and only if there exists a leaf-restricted simplification $T'$ for $T$ under  $\Gdistance(T,T')\leq 1$ with at most $k+1$ vertices. 
	
	$\Rightarrow$: Let $S= \{s_1,\dots,s_m\}$ with $m\leq k$ be a dominating set to $G$. Set $V(T')=\{w,s_1,\dots,s_k\}$, and $E(T')=\{\segment{s_iw}\mid 1\le i\le m\}$. Consider a mapping $\mu$ such that maps $\mu(w)=w$ and $\mu(p)= p'$ where $p\in V(T)$ and $p'\in V(T')$. Note that such that $p$ is contained in the unit disk centered at $p'$. Now each edge $\segment{xp}\in E(T)$ is mapped to an edge $\segment{wp'}\in E(T')$, and thus $\Gdistance(T,T')\leq 1$.
	
	$\Leftarrow$: Let $T'$ be a simplified tree of number of vertices at most $k+1$. Let $\mu$ be the mapping realizing $\Gdistance(T,T')\leq 1$. Let $C=\{\mu(p)\mid p\in V(T)\setminus\{w\}\}$, then $C\subseteq V(T')$ because $T'$ is a leaf-restricted simplification and thus $\mu$ matches each leaf $p$ of $T$ to a leaf $\mu(p)$ of $T'$. Also, $C$ contains at most $k$ vertices because $w$ is far enough from $B$, so there must be at least one vertex of $T'$ which lies outside of $B$ such that none of the points of $P$ is matched to it. Now consider the set centers $S = \{s_1,\cdots, s_m\}$ of the unit disks around $V(T')$ with $m\leq k$. Clearly, $S$ is a dominating set for $G$. The argument for traversal distance is similar. 
\end{proof}

\section{An Algorithm for Leaf-Restricted Tree-Tree Simplification from Output to Input}\label{sec:leaf-restricted-tree}
As we have seen by now, the restriction on the leaves of $T'$ and $T$ did not change the difficulty of the problem. In this section we aim to flip the direction of the distance applied between the two trees. We consider the leaf-restricted tree-tree simplification under $\Gdistance(T',T)\leq \delta$. For a given subset $l=\{l_1,\cdots, l_k\}$ of leaves  in $T$ with $k\geq 1$, we require $T'$ has $k$ leaves identical and mapped to leaves of $T$ in $l$ that are given as part of the input. Note that $T$ is rooted, so the root $r$ is given as an arbitrary vertex of $T$.
In this section, $T'$ selects its vertices from a subset of $V(T)$ along with the given leaf set $l$. 
Without any restriction on the leaves, the problem is trivial, as $T'$ could consist of a single point only.

Similar to Section~\ref{sec:FPT-tree-tree}, let $G$ be the complete graph induced by $V(T)$. 
For a vertex $u\in V(T)$, let $\I_u$ be the set of all elementary intervals $I \in \slice(u)$.
Given a leaf set $l$ in $T$, we have the following observation:

\begin{observation} \label{obs:singleleaf}
	Given a tree $T$, leaf set $l$ with root $r$ and a leaf $l_1 \in l$, a solution to the leaf-restricted min-edge tree simplification from $T'$ to $T$ when $k=1$ is a reachable path across $\FSD(G,T)$ starting from an elementary interval $I \in \slice(r)$ to some elementary interval $I' \in \slice(l_1)$. 
\end{observation}
In fact, the simplified tree for $k=1$ is trivially the path from the respective leaf to the
root. 
Our simplification algorithm simplifies the remaining vertices in $T$.
We present the algorithm below:

\paragraph{\bf The algorithm:}
Relying on Observation~\ref{obs:singleleaf}, we first compute the simple path
from every leaf in $l$ to the root in $T$. Merging these paths
forms the \emph{mapping subtree} $M_T$. A vertex where multiple paths
merge/meet in $M_T$ is called an \emph{ancestor}. 
We associate a pointer with each leaf to point to the closest
ancestor obtained along the path to $r$. We repeat the same process for the new
ancestors until we meet $r$. The ancestors and pointers together result in a
tree called \emph{ancestor tree} $A_T$. 
An ancestor $u$ is the \emph{parent} of $v$ ($v$ is a child) if it is pointed to by $v$ in $A_T$. The idea is to use
dynamic programming to propagate the optimal simplified tree rooted at
every elementary interval $I'$ of every child
to the one rooted at every elementary interval $I$ of the parent. For this, we associate a cost function $\psi:[0,1]\rightarrow \Nats$
with each edge $I$, where
$\psi(I)$ is the number of vertices in a minimum-edge simplified
subtree rooted at $I$. Similar to Section~\ref{sec:FPT-tree-tree}, we again recall that a simplified tree $T'$ \emph{is rooted at an elementary interval} $I \in \slice(u)$ if $T'$ simplifies the subtree of $T$ rooted at $u$ and  the mapping realizing $\Gdistance(T',T)\leq \delta$ matches $u$ to a point $x \in I$.  For an elementary interval $I \in \slice(u)$ and $I'\in \slice(v)$, where $u$ is the parent of $v$ in $A_T$, we consider the following recursive formula:


\begin{align}
\psi(I)& = \min\limits_{\E } \sum\limits_{I' \in \E} \big(\psi(I') + \gamma(I',I)\big), \label{formula:ancestor}
\end{align}

\begin{figure}[!t]
	\centering
	\includegraphics[width=.65\textwidth]{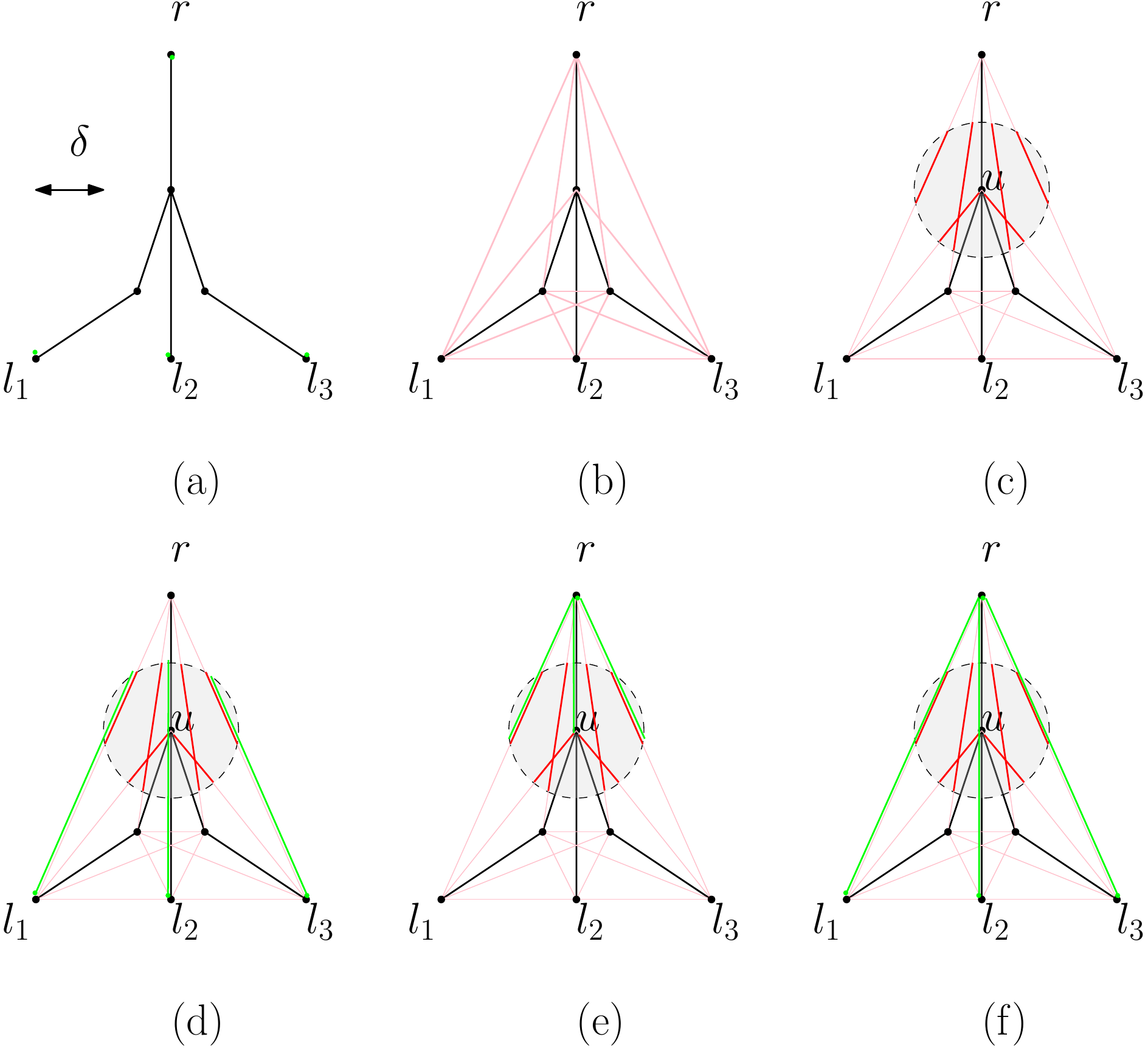}
	\caption[An illustration of the leaf-restricted simplification algorithm]{An illustration of the algorithm. 
		(a) The input tree, leaf set with $k=3$, and $\delta$.
		(b) The shortcut graph (complete graph) $G$ highlighted in pink.
		(c) Computing common ancestor $u$ of the leaves $\{l_1,l_2,l_3\}$, identifying the elementary intervals of $\slice(u)$ highlighted in red.
		(d) Minimum-vertex curve simplification of $P[l_i, u]$ highlighted in green, starting at $l_i$ and ending at all possible elementary intervals in $\slice(u)$, for all $i=1,2,3$.
		(e) The continuation of the minimum-vertex path simplification of $P[u,r]$ starting at elementary intervals in $\slice(u)$ and ending at $r$. (f) The resulting min-vertex simplified tree in green.}
	\label{fig:simpalgo}
\end{figure}

where $\E = \{\E_1 , \E_2, \cdots, \E_{k'}\}$ and each elementary interval $\E_i$ belongs to the elementary interval set $\I_{v_i}$, for all $1\leq i \leq k'$, where $k'$ is the number of children of $u$ in $A_T$. Here, $\gamma(I',I)$ is the number of spines on the reachable path in $\FSD(M_T,G)$ starting at $I$ and ending at $I'$. Note that the path from $u$ to $v_i$ in $M_T$ forms a polygonal curve, thus the simplification of the curve $P[u,v_i] \in M_T$ can be computed by the algorithm proposed in \cite{kklmw-gcs-19} (see also~\cite{kklmw-omnsp-18}). Additionally, for every leaf $u\in l$ we set $\psi(I)=1$.

\begin{remark}
	Unlike the formula in Section~\ref{sec:FPT-tree-tree}, $\E_i \cap \E_j= \emptyset$ for all $\E_i , \E_j \in \E $ with $i \neq j$. If $\E_i \cap \E_j\neq  \emptyset$ , then a point $x \in \E_i \cap \E_j$ would be mapped to two different $v_i$ and $v_j$ in $M_T$ which is not possible under the mapping realizing the graph distance from $T'$ to $T$. Our proposed algorithm can run in polynomial time relying on such a fact. Figure \ref{fig:simpalgo} illustrates on the algorithm described above. 
	
\end{remark}

\begin{lemma} \label{lem:ancestortree}
	The mapping subtree $M_T$ is the only subtree of $T$ that $T'$ is mapped to under $\Gdistance(T',T)\leq \delta$. 
\end{lemma}

\begin{proof}
	Suppose that $\mu$ is a mapping from the vertices of $T'$ to some points in $T$ realizing $\Gdistance(T',T)\leq \delta$. For the sake of contradiction, let  $M_T \subseteq T$ not be the one that $T'$ is mapped to but there be another subtree $M'_T$ where simplifying it would constitute the optimal simplification $T'$. Now there are two possible cases: (1)  $M'_T \cap M_T = \emptyset$, and (2) $M'_T \cap M_T \neq \emptyset$. 
	 In case (1), we immediately face a contradiction since $M_T$ contains all the leaves in $l$ and $M'_T$ cannot contain any of them in $l$ which is against what $\mu$ does. 

	 In case (2), there are three possible subcases: (i) $M'_T \subset M_T$, (ii) $M_T \subset M'_T$ (iii) there is some vertex $b \in M'_T$ where $b \notin M_T$. In subcase (i) correspondingly we have $V(M'_T) \subseteq V(M_T)$. In other words, there are some vertices of $M_T$ that are missing in $M'_T$. The missing vertices cannot be the root nor the leaves in $l$, because they have to be mapped from (and identical to) their corresponding vertices on $T'$.  Thus there might be some intermediate vertices in $M_T$ that are missing in $M'_T$. This implies that $M'_T$ is disconnected and  $\Gdistance(M'_T, T') = \infty$, which is a contradiction.
	 \begin{figure}[!t]
	 	\begin{center} 
	 		\includegraphics[width=0.3 \textwidth]{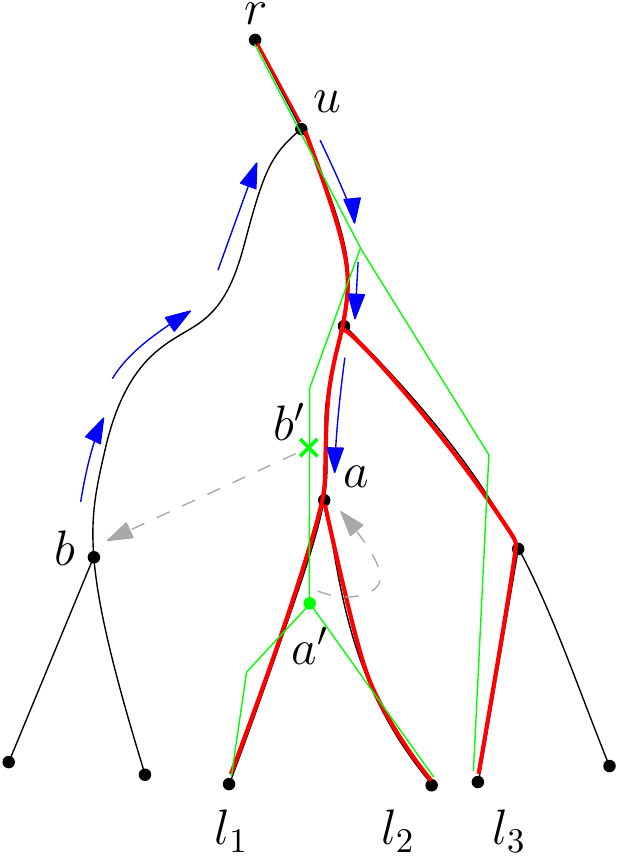}
	 	\end{center}
	 	\caption[The unique mapping subtree]{\label{fig:mapping_subtree} The input tree in solid black, the mapping subtree $M_T$ in red, the simplified tree in green are depicted above. Mapping $b'$ to $b$ results in a turn around (blue arrows) and therefore having large graph distance.}
	 \end{figure}
	 In subcase (ii),
	there is a vertex $b \in M'_T$ and $b \notin M_T$. 
	This implies that there is some point $b' \in T'$ that is mapped to $b$, i.e. $\mu(b')= b$. Suppose $b'$ is the first point  encountered along the path from some leaf to the root that $\mu(b')= b$. Now let $a' \in V(T)$ the latest vertex along the path to $b$ (the endpoint of the segment where $b'$ lies).  Note that if $b$ is a leaf then we immediately have a contradiction since $b \notin l$ due to $b \notin M_T$. Now let $u$ be the lowest ancestor of $b$ in $M_T$.  
	We necessarily know that $u \in M_T$. Observe that $\|b'-u\|>\delta$ otherwise $\mu(b')=u$ and there is no further need for assuming that $b$ exists in $M_T$. Now consider the path $\P$ in $M'_T$ from $\mu(b')$ to $\mu(a')$. Clearly $\P$ passes through $u$ (see Figure~\ref{fig:mapping_subtree}). This implies that 
	$\Frechet(\segment{b'a'},\P)>\delta$ (and $\wFrechet(\segment{b'a''},\P)>\delta$) since $\|b'-u\|>\delta$. This leads to $\Gdistance(T',T)> \delta$ (and $\WGdistance(T',T)> \delta$) which is a contradiction.
	The argument for subcase (iii) is similar to (ii).  This completes the proof.
\end{proof}

\begin{lemma}\label{lem:ancestorDPformula}
	The DP formula (\ref{formula:ancestor}) correctly computes $\psi(I)$.  
\end{lemma}
\begin{proof}
	We use a proof by induction. Suppose $I \in \slice(u)$ and $u$ is a leaf. Obviously the vertex-restricted minimum-vertex simplified tree rooted at $I$ is a single vertex tree where $u\in l$ is a leaf. Therefore $\psi(I) = 1$. 
	Now we consider the case where $u$ is an interior vertex in $T$. Suppose $T^*_I$ is an optimal simplified tree rooted at $I$. Observe that $T^*_I$ passes through a set of elementary intervals $\E = \{\E_1 , \E_2, \cdots, \E_{k'}\}$ with $\E_i \in \I_{v_i}$, for all $1\leq i \leq k'$.  Therefore we have:
	
	$$\psi(I)= \min\limits_{T'_I} \weight(T'_I) = \weight(T^*_I) = \sum\limits_{I' \in \E} (\weight(T^*_{I'})+\weight(P(I',I))),$$
	
	where $P$ is a min-edge path between $I\in \slice(u)$ and $I'\in \slice(v_i)$. Note that $\weight(P(I',I)) = \gamma(I',I)$ because otherwise $T^*_I$ would no longer be optimal and $\weight(T^*_{I'}) = \psi(I')$ by the inductive hypothesis.  
	Thus we have: 
	
	$$\psi(I)= \weight(T^*_I) = \sum\limits_{I' \in \E} \big(\psi({I'})+\gamma(I',I)\big).$$
	
	Realize that $\I_{v_i}$ consists of elementary intervals one from each child of $u$ in $A_T$. Following Lemma \ref{lem:ancestortree}, $T'$ has to be mapped to $M_T$ where $u \in A_T$ and $\Cld(u)= \{v_1,\cdots, v_{k'}\} \in A_T$. Also $T'$ has to be mapped to each of the elementary intervals in $\E = \{\E_1 , \E_2, \cdots, \E_{k'}\}$, otherwise missing one of the intervals results in $\Gdistance(T',T)> \delta$.
	Therefore, we have: 
	
	$$\psi(I)= \weight(T^*_I) = \min\limits_{\E } \sum\limits_{I' \in \E} \big(\psi({I'})+\gamma(I',I)\big),$$ as desired.
\end{proof}
\begin{theorem} \label{thm:OTTLGD}
	Let $\delta>0$ and $k>0$ be an integer. There is an algorithm running in $O(kn^5)$ time that uses $O(kn^2)$ space for the leaf-restricted tree-tree simplification under $\Gdistance(T',T)\leq \delta$.
\end{theorem}
\begin{proof}
	
	first realize that $|\I_u| = |E(G)| = O(n^2)$. Also computing $M_T$ takes $O(kn)$ time. The only remaining part is to compute $\psi(I)$ for all elementary intervals $I\in \slice(u)$ and all ancestors $u \in M_T$. Since $|\I_u|=O(n^2)$ and there are $O(k)$ vertices $u$ in $M_T$, thus there are $O(kn^2)$ starting intervals $I$ to compute $\gamma(I,I')$ for. Computing $\gamma(I,I')$ takes $O(n^3)$ for all $I'$ following the algorithm in \cite{kklmw-gcs-19} under both weak and strong Fr\'echet distances. Overall, the algorithm takes $O(kn)$+ $O(k\cdot n^2\cdot n^3)= O(kn^5)$.	Since we have $O(k)$ nodes like $u$ and $O(n^2)$ elementary intervals to store their $\psi$ values, thus the space required for this algorithm is $O(kn^2)$.
\end{proof}

\section{NP-Hardness for Edge-Restricted Simplification from Input to Output}\label{sec:NP-nonrestricted}
In this section, we show that the edge-restricted simplifications for variety of inputs and outputs to be either graph or tree or a curve, under both graph and traversal distances is (weakly) NP-hard. 
We use the NP-hardness template for edge-restricted curve simplification under the weak and strong Fr\'echet distances provided by Van Kerkhof et al. \cite{kklmw-gcs-19}. The comprehensive version of their construction is presented in \cite{kklmw-gcsarx-19}. Note that the modified construction in this section works solely under the strong graph distance, however modifying the construction in \cite{kklmw-gcsarx-19} under the weak Fr\'echet  distance in a similar way obtains us the (weakly) NP-hardness result under weak graph distance and traversal distance as well.  
\begin{figure}[!t]
	\begin{center} 
		\includegraphics[width=1.05 \textwidth]{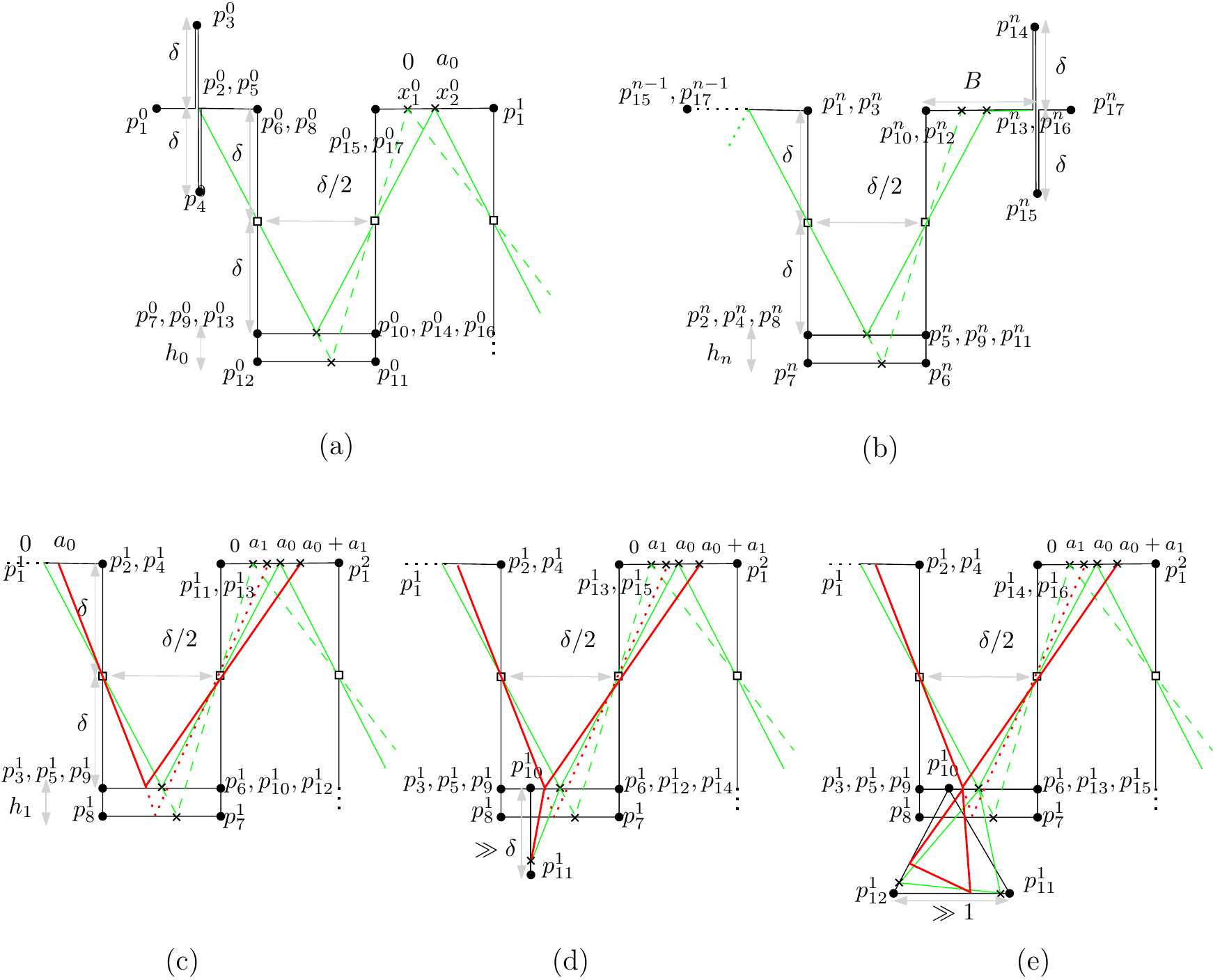}
	\end{center}
	\caption[The modified reduction tree for the edge-restricted variant]{\label{fig:edgesketch}  (a) The first modified curve gadget. (b) The last modified curve gadget. The last edge of the entire curve has length $B$. (c) The intermediate curve gadget; in this example two possible links on the simplification from the previous gadget produce four  possible links induce by the loop zone. These four links encodes all partial sums up to the intermediate gadget. (d) The intermediate tree gadget, (e) The intermediate graph gadget.}
\end{figure}

The reduction is from the subset sum problem: given a universal set $A=\{a_0, a_1, \cdots, a_n\}$ of positive integers and an integer value $B>0$, one asks for the existence of a subset of $A$ whose elements sum up to $B$. As shown in Figure~\ref{fig:edgesketch} the reduction curve in \cite{kklmw-gcs-19} is constructed in such a way that a simplified curve  $P'$ should pass through the midpoints of the vertical zigzag edges of length $2\delta$ and hit either the upper or lower edge in the loop zone at the bottom of each gadget. The first and last gadgets have the zigzag edges of length $2\delta$ on their top left and right edges, respectively. These zigzag edges control the simplified curve to start and end at an edge on the input curve and not necessarily at a vertex. In the last gadget the the zigzag edge is located at distance $B$ to the previous vertex on the top right edge as shown in Figure~\ref{fig:edgesketch} (b). The horizontal distance between the first vertex and the zigzag edges in the first gadget is sufficiently small. Similarly the distance between the last vertex and the zigzag edges in the last gadget is small enough as well.

The entire curve is a sequential combination of the first gadget $g_0$, $n-1$ similar intermediate gadgets $g_i$ for $0<i<n$, and the last gadget $g_n$, i.e., $ \langle g_0,  g_1, \cdots, g_n\rangle $, that are arranged rightward. The vertex $p^i_j$ is the $j$th vertex of $g_i$ in our construction with $0 \leq i\leq n$ and $1\leq j \leq 17$.
The gadget $g_i$ takes the integer $a_i \in A$ and creates a loop zone of height $h_i=\frac{a_i\delta}{\delta/2-a_i}$. In this construction along with the one presented in~\cite{kklmw-gcsarx-19} it is assumed that $\delta > \sum_{i=0}^{n} a_i$.  Note that the \emph{``loop zone''} is only a simple path whose vertices are overlaid onto top of each other. It then produces a set  $S_i $ of all partial sums of $A_i = \{a_0, a_1,\cdots, a_i\}$ on the top right horizontal edge of the curve demonstrated by $X = \{x^i_1,\cdots , x^i_{2^{i+1}}\}$. In other words, the difference between the points in $X = \{x^i_1,\cdots , x^i_{2^{i+1}}\}$ induced by the two possible links (the solid and dashed one) hitting the top right edge produces a partial sum involving the new integer $a_i$ in the set $A$.  This way, when $g_{i+1}$ takes $a_{i+1}$, it produces the set $S_{i+1}= \{s_i + a_{i+1} ~|~s_i \in S_i\}$, where $S_i$ is the set of all partial sums encoded in $g_i$. In $g_n$ the length of top right horizontal edge is equal to $B$. This way, if there is a partial sum $s_n = B$, where $s_n \in S_N$, then there is a minimum-edge curve simplification of number of edges at most $2(n+1)$. 

Now all we need is to extend the constructed curve gadgets (Figure~\ref{fig:edgesketch} (a), (b), (c))  to a tree (or graph). In order to do this we add an edge hung from the loop zone at the bottom whose length is greater than $\delta$ in a way that each simplified tree hitting either the top or bottom edge of the loop, should also fall within the ball of radius $\delta$ around the bottom most point of the edge. Thus, it needs to use one more edge towards the bottom most vertex (Figure~\ref{fig:edgesketch} (d)). 

This way we only need 3 edges per gadget and the decision parameter on the number of edges is set $3(n+1)$. In the case that we want to simplify a graph with a graph of minimum complexity, all we have to do is to build another loop to force the simplified object contains a loop. The auxiliary loop hung from the loop zone has to have diameter larger than $\delta$ consisting of 3 edges. The decision parameter in this case is $5(n+1)$. We have the following theorem:

\begin{theorem} \label{thm:edge-restricted-gg}
	The min-complexity edge-restricted graph-graph, and the min-edge edge-restricted tree-tree from input to output are (weakly) NP-hard.
\end{theorem}

\section{Concluding Remarks}
In this paper, we studied the problem of approximating a graph with an alternative simpler graph with a minimum complexity preserving the Fr\'echet-like distances between them. To this end, we considered the two main Fr\'echet-like distances; traversal and graph distances under different constraints in which the vertices of the simplified graph can be placed. While this was an initial work under such distances, we obtained a set of NP-hardness and algorithmic results depending on the problem variants and have left some of the variants as open problems. We believe that other variants of the problem that we have not covered in this paper, and have applications in real life, may admit polynomial-time algorithms. Any further investigation on this problem, providing approximation algorithms  and obtaining a new results for the non-restricted case can be of interest.


\bibliographystyle{abbrv}
\bibliography{treesimplification}

\providecommand{\NOOPSORT}[1]{}
\begin{thebibliography}{10}

\bibitem{abksw-19}
H.~Akitaya, M.~Buchin, B.~Kilgus, S.~Sijben, and C.~Wenk.
\newblock Distance measures for embedded graphs.
\newblock {\em Computational Geometry, Theory And Applications},
  95(101743):1--21, 2021.

\bibitem{aerw-mpm-03}
H.~Alt, A.~Efrat, G.~Rote, and C.~Wenk.
\newblock Matching planar maps.
\newblock {\em Journal of Algorithms}, 49(2):262--–283, 2003.

\bibitem{ag-cfdb-95}
H.~Alt and M.~Godau.
\newblock Computing the {F}r\'{e}chet distance between two polygonal curves.
\newblock {\em International Journal of Computational Geometry and
  Applications}, 5(1--2):75--91, 1995.

\bibitem{bs-rnaspsaknn-06}
E.~Bindewald and B.~Shapiro.
\newblock {RNA} secondary structure prediction from sequence alignments using a
  network of k-nearest neighbor classifiers.
\newblock {\em RNA}, 12(3):342--352, 2006.

\bibitem{bos-ctd-17}
K.~Buchin, T.~Ophelders, and B.~Speckmann.
\newblock Computing the {F}r{\'e}chet distance between real-valued surfaces.
\newblock In {\em Proceedings of the 2017 Annual ACM-SIAM Symposium on Discrete
  Algorithms}, SODA'17, pages 2443--2455, 2017.

\bibitem{cgkss-msgg-08}
O.~Cheong, J.~Gudmundsson, H.-S. Kim, D.~Schymura, and F.~Stehn.
\newblock Measuring the similarity of geometric graphs.
\newblock In J.~Vahrenhold, editor, {\em Experimental Algorithms}, pages
  101--112, Berlin, Heidelberg, 2009. Springer Berlin Heidelberg.

\bibitem{cms-cmsa-98}
P.~Cignoni, C.~Montani, and R.~Scopigno.
\newblock A comparison of mesh simplification algorithms.
\newblock {\em Computers \& Graphics}, 22(1):37--54, 1998.

\bibitem{ccj-udg-90}
B.~Clark, C.~Colbourn, and D.~Johnson.
\newblock Unit disk graphs.
\newblock {\em Discrete Math.}, 86(1--3):165--177, 1990.

\bibitem{estkowski01simple}
R.~Estkowski and J.~S.~B. Mitchell.
\newblock Simplifying a polygonal subdivision while keeping it simple.
\newblock In {\em Proceedings 17th Annual {ACM} Symposium on Computational
  Geometry}, SCG '01, pages 40--49, 2001.

\bibitem{funke2017map}
S.~Funke, T.~Mendel, A.~Miller, S.~Storandt, and M.~Wiebe.
\newblock Map simplification with topology constraints: Exactly and in
  practice.
\newblock In {\em Proc. 19th Workshop on Algorithm Engineering and Experiments
  (ALENEX)}, pages 185--196, 2017.

\bibitem{glw-mpstd-07}
J.~Gudmundsson, P.~Laube, and T.~Wolle.
\newblock Movement patterns in spatio-temporal data.
\newblock In S.~Shekhar and H.~Xiong, editors, {\em Encyclopedia of GIS}.
  Springer-Verlag, 2007.

\bibitem{ghms-apswmlp-93}
L.~Guibas, J.~Hershberger, J.~Mitchell, and J.~Snoeyink.
\newblock Approximating polygons and subdivisions with minimum-link paths.
\newblock {\em International Journal of Computational Geometry \&
  Applications}, 3(4):383--415, 1993.

\bibitem{jh-ablpfged-06}
D.~{Justice} and A.~{Hero}.
\newblock A binary linear programming formulation of the graph edit distance.
\newblock {\em IEEE Transactions on Pattern Analysis and Machine Intelligence},
  28(8):1200--1214, 2006.

\bibitem{kklmw-omnsp-18}
M.~{\NOOPSORT{Kerkhof}}van~de Kerkhof, I.~Kostitsyna, M.~L{\"o}ffler,
  M.~Mirzanezhad, and C.~Wenk.
\newblock On optimal min-\# curve simplification.
\newblock In {\em 28th Fall Workshop on Computational Geometry}, (FWCG 2018),
  2018.

\bibitem{kklmw-gcs-19}
M.~{\NOOPSORT{Kerkhof}}van~de Kerkhof, I.~Kostitsyna, M.~L{\"o}ffler,
  M.~Mirzanezhad, and C.~Wenk.
\newblock Global curve simplification.
\newblock In {\em 27th Annual European Symposium on Algorithms (ESA 2019)},
  volume 144, pages 1--14, Dagstuhl, Germany, 2019.

\bibitem{kklmw-gcsarx-19}
M.~{\NOOPSORT{Kerkhof}}van~de Kerkhof, I.~Kostitsyna, M.~L{\"o}ffler,
  M.~Mirzanezhad, and C.~Wenk.
\newblock Global curve simplification.
\newblock \url{http://arxiv.org/abs/1809.10269}, 2019.

\bibitem{l-hdrknn-91}
Y.~Lee.
\newblock Handwritten digit recognition using k nearest-neighbor, radial-basis
  function, and backpropagation neural networks.
\newblock {\em Neitrul Computation}, 3(3):440--449, 1991.

\bibitem{m-apsstcep-18}
T.~Mendal.
\newblock Area-preserving subdivision simplification with topology constraints:
  Exactly and in practice.
\newblock In {\em Proc. 20th Workshop on Algorithm Engineering and Experiments
  (ALENEX)}, pages 117--128, 2018.

\bibitem{m-smcs-14}
W.~Meulemans.
\newblock {\em Similarity measures and algorithms for cartographic
  schematization}.
\newblock PhD thesis, Eindhoven University of Technology, Eindhoven University
  of Technology, 2014.

\end{thebibliography}

\end{document}